\documentclass[12pt]{article}
\usepackage{eurosym}
\usepackage{float}
\usepackage{makeidx}
\usepackage{graphicx}
\usepackage{amsmath}
\usepackage{authblk}
\usepackage[dcucite]{harvard}

\newtheorem{Theo}{Theorem}

\newtheorem{Corol}{Corollary}
\newtheorem{Def}{Definition}
\newtheorem{Le}{Lemma}

\newtheorem{Rem}{Remark}
\newtheorem{Exp}{Example}
\input{tcilatex}
\begin{document}

\date{}
\title{\textbf{Application of the unified control and detection framework to
detecting stealthy integrity cyber-attacks on feedback control systems }}
\author[1]{\small Steven X. Ding}
\author[2]{\small Linlin Li}
\author[1]{\small Dong Zhao}
\author[1]{\small Chris Louen}
\author[1]{\small Tianyu Liu}

\affil[1]{\small Institute for Automatic Control and Complex Systems, University of Duisburg-Essen, 47057, Duisburg, Germany}
\affil[2]{School of Automation and Electrical Engineering, University of Science and Technology Beijing, Beijing 100083, P. R. China (Corresponding author)}

\maketitle

\bigskip

\textbf{Abstract}: This draft addresses issues of detecting stealthy
integrity cyber-attacks on automatic control systems in the unified control
and detection framework. A general form of integrity cyber-attacks that
cannot be detected using the well-established observer-based technique is
first introduced as kernel attacks. The well-known replay, zero dynamics and
covert attacks are special forms of the kernel attacks. Existence conditions
for the kernel attacks are presented. It is demonstrated, in the unified
framework of control and detection, that all kernel attacks can be
structurally detected when not only the observer-based residual, but also
the control signal based residual signals are generated and used for the
detection purpose. Based on the analytical results, two schemes for
detecting the kernel attacks are then proposed, which allow reliable attack
detection without loss of control performance. While the first scheme is
similar to the well-established moving target method and auxiliary system
aided detection scheme, the second detector is realised with encrypted
transmissions of control and monitoring signals in the feedback control
system that prevents adversary to gain system knowledge by means of
eavesdropping attacks. Both schemes are illustrated by examples of detecting
replay, zero dynamics and covert attacks and an experimental study on a
three-tank control system.

\bigskip


\section{Introduction}

Automatic control systems are essential system parts of many industrial
cyber-physical systems (CPSs) and their flawless operations are of elemental
importance for optimal system operation and high product quality. It is
therefore not surprising that automatic control systems are often immediate
targets of cyber-attacks on industrial CPSs. Driven by the rapidly
increasing industrial demands for higher cyber-security, detection of
cyber-attacks on automatic control systems has drawn incredible research
attention in the current decade. Excellent reviews of state of the art of
research in this thematic area can be found in the recent surveys published
in \cite%
{DHXGZ2018,Survey-attack-detection2018,DIBAJI2019-survey,YMA2019,TGXHV2020,ZHANG2021,Zhou2021IEEE-Proc}%
.

\bigskip

Among various types of cyber-attacks, integrity attacks are specially
directed to automatic control systems \cite{DIBAJI2019-survey,GWSOM2019}. By
injecting attack signals into system input and output channels, e.g. via I/O
and network interfaces, integrity attacks can lead to remarkable system
performance degradations and even catastrophic damages. An early and
reliable detection of integrity attacks is becoming a vital requirement on
cyber-security of industrial CPSs, for instance, for power control systems 
\cite{MMM2020}. Thanks to its well-established theoretical framework in the
past three decades, observer-based fault detection technique \cite{Ding2013}
is widely accepted as an efficient method, among numerous ones, to deal with
detection of integrity attacks on control systems \cite%
{DIBAJI2019-survey,GWSOM2019,TGXHV2020}. Unfortunately, different from
technical faults, cyber-attacks are artificially created and can be designed
and generated by an adversary. It is particularly insidious, when
cyber-attacks are generated in such a way that they cannot be detected using
the known detection techniques. Such cyber-attacks are called stealthy. This
observation and some real examples with stealthy cyber-attacks have strongly
motivated researchers to improve the existing detection schemes and develop
alternative solutions. In this regard, a great number of results have been
reported about detecting the so-called replay, zero dynamics and covert
attacks, which are stealthy integrity attacks as the standard observer-based
fault detection technique cannot detect them without modifications on the
applied algorithms \cite{DIBAJI2019-survey,GWSOM2019,TGXHV2020}.
Representative solutions are the watermark detection scheme \cite%
{Mo2015-Watermarked-detection}, the moving target method \cite%
{MT-method-CDC2015} and the auxiliary system aided detection scheme \cite%
{Zhang-CDC2017}, just citing the initial works on these methods. Our work is motivated
by the above observation and in particular driven by the questions like: what
is the general form of stealthy integrity attacks? what are the existence
conditions for such stealthy integrity attacks? is it possible to develop a
general observer-based scheme applied to detecting integrity attacks in
automatic control systems? Satisfactory answers to these questions could
help us (i) to reveal possible weakness of observer-based detection
technique by dealing with integrity cyber-attacks, and thus (ii) to prevent
new variations of stealthy integrity attacks, and (iii) to develop new
detection schemes, in particular such ones that are able to detect major
types of integrity attacks. The main objective of our work is to investigate
possible answers to the above questions. Different from the reported
studies, our work will study the issues of stealthy integrity attacks in the
unified framework of control and detection.

\bigskip

Inspired by the work in \cite{ZR2001}, and based on the parameterisations of
observers and observer-based residual generators \cite{Ding2013}, Ding et
al., 2010 \nocite{DYZDJWS2009}proposed an observer-based realisation and
implementation of all stabilising (dynamic output) controllers whose core is
an observer-based residual generator, and demonstrated its successful
applications. In the recent decade, on the basis of this work, a new unified
framework of control and detection has been established, which generalises
the integrated design schemes for control and detection initiated by Nett et
al. (1988) \nocite{NJM88} and further developed in the past decades \cite%
{KRNS96,SGN97,KNS04,HenryAUTO05,WY-LPV-08}. It has been applied to fault
diagnosis in automatic control systems with uncertainties, fault-tolerant
control and, more recently, to control performance degradation monitoring,
detection and recovery \cite{Ding2020}. The basic idea behind the control
and detection unified framework is that any controller is indeed
residual-driven and can be implemented in form of an observer and an
observer-based residual generator. This allows to extend the residual-based
detection space to the overall measurement space spanned by the system
inputs and outputs. As a result, it can be expected that the system
capability for detecting cyber-attacks is (considerably) enhanced.

\bigskip

The intended contributions of our work are summarised as

\begin{itemize}
\item revealing that any attacks lying in the system kernel space cannot be
detected by an observer-based detection system. In this context, the concept
of kernel attacks is introduced, which provides us with a general expression
of all stealthy integrity attacks (with respect to the observer-based
detection technique);

\item presenting existence conditions that integrity attacks are stealthy in
the unified framework of control and detection, and based on them,

\item proposing two schemes for detecting the kernel attacks (thus
including detecting replay, zero dynamics and covert attacks). The first one
is a natural extension of the observer-based detection schemes to a unified
control and detection system, while the second one is dedicated to a
detection scheme with encrypted transmissions of control and monitoring
signals in the feedback control system under consideration. This is helpful
to prevent adversary to gain system knowledge by means of eavesdropping
attacks.
\end{itemize}

\bigskip

The paper is organised as follows. In Section 2, the unified framework of
control and detection is first presented together with the necessary control
theoretical and mathematical preliminaries. It is followed by a short review
of replay, zero dynamics and covert attacks. Section 3 is dedicated to the
study on stealthy integrity attacks and introduction of the concept of
kernel attacks as a general form of stealthy integrity attacks. In Section
4, existence conditions for stealthy integrity attacks are first investigated and
presented. They build the basis for the development of two schemes for
detecting kernel attacks. These two schemes are presented in Sections 4 and
5, respectively. Their capability of detecting the kernel attacks are
illustrated and demonstrated by examples and experimental results in Section
6.

\bigskip

Throughout this paper, standard notations known in linear algebra and
advanced control theory are adopted. In addition, $\mathcal{RH}_{\infty }$
is used to denote the set of all stable systems. In the context of
cyber-attacks, when signal $\xi $ is attacked, it is denoted by $\xi ^{a},$
and the corresponding (injected) attack signal by $a_{\xi },$ i.e. $\xi
^{a}=\xi +a_{\xi }$.

\section{Preliminaries of system models, the unified framework of control
and detection, and stealthy integrity attacks}

As the methodological basis of our work, we first introduce the unified
framework of control and detection. It is followed by a short review of
system descriptions of stealthy integrity cyber-attacks on feedback control
systems.

\subsection{System representations and controller parameterisation}

\subsubsection{System factorisations, observer-based residual generation and
kernel space}

Consider a nominal plant model%
\begin{equation}
y(z)=G_{u}(z)u(z),y(z)\in \mathcal{C}^{m},u(z)\in \mathcal{C}^{p}
\label{eq2-1}
\end{equation}%
with $u$ and $y$ as the plant input and output vectors. It is assumed that $%
G_{u}(z)$ is a proper real-rational matrix and its minimal state space
realisation is given by the following discrete-time linear time invariant
(LTI) system%
\begin{align}
x(k+1)& =Ax(k)+Bu(k),x(0)=x_{0},  \label{eq2-2a} \\
y(k)& =Cx(k)+Du(k),  \label{eq2-2b}
\end{align}%
where $x\in \mathcal{R}^{n}$ is the state vector and $x_{0}$ is the initial
condition of the system. Matrices $A,B,C,D$ are appropriately dimensioned
real constant matrices. A coprime factorisation of a transfer function
matrix over $\mathcal{RH}_{\infty }$ gives a further system representation
form and factorises the transfer matrix into two stable and coprime transfer
matrices. The left and right coprime factorisations (LCF and RCF) of $%
G_{u}(z)$ are given by 
\begin{equation}
G_{u}(z)=\hat{M}^{-1}(z)\hat{N}(z)=N(z)M^{-1}(z),  \label{eq2-3}
\end{equation}%
where the state space realisations of the left and right coprime pairs (LCP
and RCP) $\left( \hat{M}(z),\hat{N}(z)\right) $ and $\left( M(z),N(z)\right) 
$ are 
\begin{align}
\hat{M}(z)& =\left( A-LC,-L,C,I\right) ,\hat{N}(z)=\left(
A-LC,B-LD,C,D\right) ,  \label{eq2-4a} \\
M(z)& =\left( A+BF,B,F,I\right) ,N(z)=\left( A+BF,B,C+DF,D\right) .
\label{eq2-4b}
\end{align}%
Correspondingly, there exist RCP and LCP $\left( \hat{X}(z),\hat{Y}%
(z)\right) $ and $\left( X(z),Y(z)\right) $ so that the so-called Bezout
identity holds%
\begin{equation}
\left[ 
\begin{array}{cc}
X(z) & \text{ }Y(z) \\ 
-\hat{N}(z) & \text{ }\hat{M}(z)%
\end{array}%
\right] \left[ 
\begin{array}{cc}
M(z) & \text{ }-\hat{Y}(z) \\ 
N(z) & \text{ }\hat{X}(z)%
\end{array}%
\right] =\left[ 
\begin{array}{cc}
I\text{ } & 0\text{ } \\ 
0\text{ } & I\text{ }%
\end{array}%
\right] .  \label{eq2-5}
\end{equation}%
The state space computation formulas for $\left( \hat{X}(z),\hat{Y}%
(z)\right) $ and $\left( X(z),Y(z)\right) $ are%
\begin{align}
\hat{X}(z)& =\left( A+BF,L,C+DF,I\right) ,\hat{Y}(z)=\left(
A+BF,-L,F,0\right) ,  \label{eq2-6a} \\
X(z)& =\left( A-LC,-(B-LD),F,I\right) ,Y(z)=\left( A-LC,-L,F,0\right) .
\label{eq2-6b}
\end{align}%
In (\ref{eq2-4a})-(\ref{eq2-6b}), (real) matrices $F$ and $L$ are selected
such that $A+BF$ and $A-LC$ are Schur matrices \cite{Zhou98,Ding2014}.

\bigskip

We now consider an observer-based residual generator 
\begin{align}
\hat{x}(k+1)& =A\hat{x}(k)+Bu(k)+L\left( y(k)-\hat{y}(k)\right) ,
\label{eq2-7a} \\
r_{0}(k)& =y(k)-\hat{y}(k),\hat{y}(k)=C\hat{x}(k)+Du(k)  \label{eq2-7b}
\end{align}%
with $r_{0}(k)$ being the primary form of a residual vector. It can be
equivalently written as%
\begin{gather}
\hat{x}(k+1)=\left( A-LC\right) \hat{x}(k)+\left( B-LD\right) u(k)+Ly(k), 
\notag \\
\Longrightarrow r_{0}(z)=y(z)-\hat{y}(z)=\hat{M}(z)y(z)-\hat{N}(z)u(z).
\label{eq2-7}
\end{gather}%
Note that if there exists no uncertainty in the plant and $x(0)=\hat{x}(0),$
it holds 
\begin{equation*}
r_{0}(z)=0\Longrightarrow y(z)=\hat{M}^{-1}(z)\hat{N}(z)u(z),
\end{equation*}%
which illustrates the interpretation of LCF as an observer-based residual
generator. It is well-known that given plant model (\ref{eq2-1}), all LTI
residual generators can be parameterised by%
\begin{equation}
r(z)=R(z)r_{0}(z)=R(z)\left( y(z)-\hat{y}(z)\right) ,R(z)\in \mathcal{RH}%
_{\infty },  \label{eq2-8}
\end{equation}%
where $R(z)$ is the parameterisation transfer function matrix \cite{Ding2013}%
.

\begin{Rem}
Hereafter, we may drop out the domain variable $z$ or $k$ when there is no
risk of confusion.
\end{Rem}

\subsubsection{Parameterisation of stabilising controllers and basics of the
unified control and detection framework}

Consider the feedback control loop 
\begin{equation*}
y(z)=G_{u}(z)u(z),u(z)=K(z)y(z)
\end{equation*}%
with the plant model $G_{u}(z)$ and controller $K(z).$ It is a well-known
result that all stabilising controllers can be parameterised by%
\begin{align}
K(z)& =-\left( X(z)-Q(z)\hat{N}(z)\right) ^{-1}\left( Y(z)+Q(z)\hat{M}%
(z)\right)  \label{eq2-12a} \\
& =-\left( \hat{Y}(z)+M(z)Q(z)\right) \left( \hat{X}(z)-N(z)Q(z)\right) ^{-1}
\label{eq2-12b}
\end{align}%
with the parameter system $Q(z)\in \mathcal{RH}_{\infty },$ where the four
coprime pairs $\left( \hat{M},\hat{N}\right) ,$ $\left( M,N\right) ,$ $%
\left( \hat{X},\hat{Y}\right) $ and $\left( X,Y\right) $ are given in (\ref%
{eq2-4a})-(\ref{eq2-6b}) and satisfy Bezout identity (\ref{eq2-5}). The
parameterisation expression (\ref{eq2-12a})-(\ref{eq2-12b}) is called Youla
parameterisation \cite{Zhou98}. It follows from (\ref{eq2-4a})-(\ref{eq2-6b}%
) and Bezout identity \cite{DYZDJWS2009,Ding2020} that any (stabilising)
output feedback controller 
\begin{equation}
u(z)=K(z)y(z)+v(z)  \label{eq2-12c}
\end{equation}%
with $v(z)$ being the reference signal can be equivalently written as%
\begin{align}
\hat{x}(k+1)& =A\hat{x}(k)+Bu(k)+Lr_{0}(k),  \label{eq2-13a} \\
u(z)& =F\hat{x}(z)-Q(z)r_{0}(z)+\bar{v}(z),  \label{eq2-13b} \\
\bar{v}(z)& =\left( X(z)-Q(z)\hat{N}(z)\right) v(z).  \label{eq2-13c}
\end{align}%
In other words, any output feedback controller is an observer-based
controller and driven by the residual signal $r_{0}$.

\bigskip

Recall that the basis of an observer-based fault diagnosis is residual
generation and evaluation \cite{Ding2013}. Thus, (\ref{eq2-13a})-(\ref%
{eq2-13b}) reveal that both diagnosis and control are driven by the residual
signal and can be integratedly realised by sharing a common observer-based
residual generator as the information provider. By means of the observer
parameterisation \cite{Ding2013}, we gain a deeper insight into the
information aspect of a feedback controller that the control signal $u(k)$ in (\ref%
{eq2-13b}) is an estimate for $Fx(k)+\bar{v}(k)$ and satisfies 
\begin{equation}
\forall x(0),u(k),\text{ }\lim\limits_{k\rightarrow \infty }\left(
u(k)-Fx(k)-\bar{v}(k)\right) =0,  \label{eq2-15}
\end{equation}%
when there exists no uncertainty in the plant. The observer-based
realisation of stabilising feedback controllers (\ref{eq2-13a})-(\ref%
{eq2-13b}) and the estimator interpretation (\ref{eq2-15}) of (any) output
feedback controllers are the basics of the unified control and detection
framework and build the basis for our study on attack detection schemes
presented in the subsequent work.

\subsection{Integrity attacks under consideration}

The system configuration under consideration in\ the first part of our study
is sketched in Figure 1, in which the controller and attack detection system
are networked with the plant (equipped with sensors, actuators and a
computation system like micro-controllers). Via the communication network,
the plant receives the control signal $u(k)$ and sends the sensor signal $%
y(k)$ to the control and monitoring system.

\begin{figure}[h]
\centering\includegraphics[width=11cm,height=5.5cm]{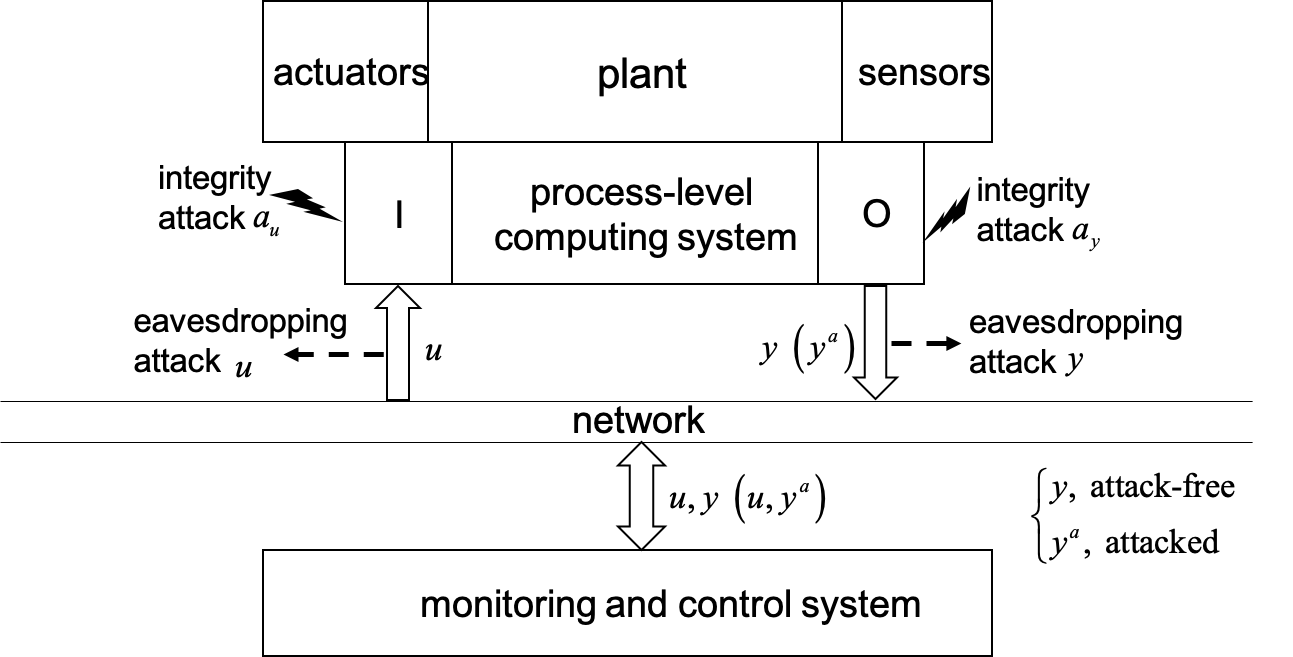}
\caption{System configuration under consideration}
\end{figure}

\bigskip
Recall that our major attention is paid to integrity cyber-attacks that are
injected into the system I/O interface via the network, cause (considerable)
changes in the system dynamics, but cannot be detected by a standard
observer-based detector. As reviewed in \cite{DIBAJI2019-survey}, such
cyber-attacks include zero dynamics, covert and replay attacks. Below is a
short description of these attack types.

\subsubsection{\textbf{Zero dynamics attacks}}

Roughly speaking, a zero dynamics attack is referred to an attack $a_{u}(k)$
on the actuators, which causes no response at the system output over the
detection time interval $\left[ k_{0},k_{0}+N\right] $ and thus cannot be
detected \cite{TEIXEIRA-zero-attack_2015}. The corresponding attack model is 
\begin{align}
x(k+1)& =Ax(k)+B\left( u(k)+a_{u}(k)\right) , \\
y^{a}(k)& =Cx(k)+D\left( u(k)+a_{u}(k)\right)
\end{align}%
with $y^{a}(k)$ satisfying the condition 
\begin{equation}
\forall k\in \left[ k_{0},k_{0}+N\right] ,y^{a}(k)=y(k).  \label{eq2-11}
\end{equation}%
It is obvious that the existence condition of zero dynamics attacks can be
expressed by means of the LCF of the plant as%
\begin{equation}
\hat{N}\left( z\right) a_{u}(z)=0.  \label{eq2-11b}
\end{equation}

\subsubsection{\textbf{Covert attacks}}

Introduced by \cite{Smith2015}, covert attacks\textbf{\ }are modelled by%
\textbf{\ } 
\begin{align}
x(k+1)& =Ax(k)+B\left( u(k)+a_{u}(k)\right) , \\
y^{a}(k)& =Cx(k)+D\left( u(k)+a_{u}(k)\right) +a_{y}(k)
\end{align}%
with $a_{u}(k)$ and $a_{y}(k)$ denoting attacks on the actuators and
sensors, respectively, and satisfying 
\begin{equation}
a_{y}(z)+G_{u}(z)a_{u}(z)=0\Longrightarrow \forall k\in \left[ k_{0},k_{0}+N%
\right] ,y^{a}(k)=y(k).  \label{eq2-17}
\end{equation}%
It is straightforward that the existence condition for covert attacks is%
\begin{equation}
\hat{M}(z)a_{y}(z)+\hat{N}\left( z\right) a_{u}(z)=0.  \label{eq2-17a}
\end{equation}

\subsubsection{\textbf{Replay attacks}}

As described in \cite{Mo2015-Watermarked-detection}, replay attacks are
performed on the assumption that the plant under attacks is operating in the
steady state, which yields%
\begin{equation*}
y(k)\approx y(k-i),i=1,\cdots .
\end{equation*}%
Consequently, the attacker can \textquotedblleft replay", e.g. over the time
interval $\left[ k,k+M\right] ,$ the sensor data collected in the past (for
instance, by means of an eavesdropping attack), and simultaneously inject
signals in the actuators. Denote by $y\left( k_{0}+i\right)
,k_{0}+M<k,i=0,1,\cdots ,M,$ the data collected and saved by the attacker.
Replay attacks can be modelled by\textbf{\ } 
\begin{align*}
x(j+1)& =Ax(j)+B\left( u(j)+a_{u}(j)\right) ,j\in \left[ k,k+M\right] , \\
y^{a}(j)& =Cx(j)+D\left( u(j)+a_{u}(j)\right) +a_{y}(j), \\
a_{y}(j)& =y(k_{0}+j-k)-\left( Cx(j)+D\left( u(j)+a_{u}(j)\right) \right) .
\end{align*}%
As a result, 
\begin{equation}
\forall j\in \left[ k,k+M\right] ,y^{a}(j)=y(j-k+k_{0})\approx
y(j)=Cx(j)+Du(j).  \label{eq2-20}
\end{equation}%
Notice that the attack signal $a_{y}(j)$ depends on the plant state vector $%
x(j)$ and is, therefore, not a pure additive attack signal.

\bigskip

In summary, it can be seen that the above three types of attacks have one
thing in common that they do not cause changes in the measurement output and
hence cannot be traced by the output variables. As a result, these attacks
cannot be detected using an observer-based detection scheme. In this
context, they are called stealthy attacks \cite{DIBAJI2019-survey}.

\subsection{Problem formulation}

The goal of our work is to investigate cyber-attacking issues in the unified
control and detection framework. We will first deal with the following three
problems:

\begin{itemize}
\item study on general system structural conditions, under which the
above-mentioned three types of attacks cannot be detected using an
observer-based detector. Based on the achieved results, a general class of
stealthy integrity cyber-attacks, the so-called kernel cyber-attacks, are
then defined;

\item derivation of system structural conditions, under which any integrity
cyber-attacks, as sketched in Figure 1, can be (structurally) uniquely
detected, and based on them,

\item development of an alternative attack detection scheme that ensures a
reliable detection of the integrity cyber-attacks shown in Figure 1.
\end{itemize}

A major reason why an integrity cyber-attack could be performed stealthily
is that the attacker has knowledge of system dynamics. One potential tool to
gain such knowledge is to collect sufficient plant input and output data by
means of eavesdropping attacks, which enable, for instance, the
identification of the plant model and even controller parameters. Under this
consideration, we will, in the further part of our work, propose an
alternative system configuration that leads to an encrypted data
transmission aiming at preventing attackers to gain system knowledge.

\section{Kernel attacks: a general form of stealthy integrity attacks}

In this section, we investigate the existence conditions for stealthy
attacks and generalise the different types of stealthy integrity attacks,
including the three types of integrity attacks introduced in the previous
section, as the so-called kernel attacks. To this end, we consider, in the
sequel, the system configuration sketched in Figure 1.

\subsection{Observer-based attack detection Strategy}

For our purpose, we extend the nominal model (\ref{eq2-1})-(\ref{eq2-2b}) to
the following attack model, 
\begin{align}
x\left( k+1\right) & =Ax\left( k\right) +B\left( u(k)+a_{u}(k)\right)
+\omega (k),  \label{eq3-1a} \\
y^{a}(k)& =Cx(k)+D\left( u(k)+a_{u}(k)\right) +a_{y}(k)+\nu (k),
\label{eq3-1b}
\end{align}%
where $\omega (k),\nu (k)$ represent the process and measurement noise
vectors, and $a_{y}(k),a_{u}(k)$ denote the attack signals on the actuators
and sensors, respectively. With respect to the system configuration shown in
Figure 1, an observer-based attack detector consists of (i) a residual
generator as given in (\ref{eq2-7a})-(\ref{eq2-7b}) with the generated
residual vector $r_{0}(k),$%
\begin{equation*}
r_{0}(k)=y^{a}(k)-\hat{y}^{a}(k),\hat{y}^{a}(k)=C\hat{x}(k)+Du(k),
\end{equation*}%
(ii) a residual evaluation function 
\begin{equation*}
J(k)=J\left( \left\Vert r_{0}(k)\right\Vert \right)
\end{equation*}%
with $\left\Vert r_{0}(k)\right\Vert $ denoting a certain norm of $r_{0}(k),$
and (iii) detection logic described by%
\begin{equation*}
\left\{ 
\begin{array}{l}
J(k)\leq J_{th}\Longrightarrow \text{attack-free,} \\ 
J(k)>J_{th}\Longrightarrow \text{attack is detected,}%
\end{array}%
\right.
\end{equation*}%
where $J_{th}$ is the threshold. In order to achieve an optimal attack
detection, the observer gain matrix $L$, the evaluation function $J(k)$ and
the threshold $J_{th}$ are designed taking into account of the statistic
properties of $\omega (k),\nu (k).$ Suppose that $\omega (k),\nu (k)$ are
uncorrelated with the state and input vectors and satisfy 
\begin{gather}
\omega (k)\sim \mathcal{N}\left( 0,\Sigma _{\omega }\right) ,\nu (k)\sim 
\mathcal{N}\left( 0,\Sigma _{\nu }\right) ,x\left( 0\right) \sim \mathcal{N}%
\left( 0,\Pi _{0}\right) ,  \label{eq3-2a} \\
\mathcal{E}\left( \left[ 
\begin{array}{c}
\omega (i) \\ 
\nu (i) \\ 
x\left( 0\right)%
\end{array}%
\right] \left[ 
\begin{array}{c}
\omega (j) \\ 
\nu (j) \\ 
x\left( 0\right)%
\end{array}%
\right] ^{T}\right) =\left[ 
\begin{array}{cc}
\left[ 
\begin{array}{cc}
\Sigma _{\omega } & S \\ 
S^{T} & \Sigma _{\nu }%
\end{array}%
\right] \delta _{ij} & 0 \\ 
0 & \Pi _{0}%
\end{array}%
\right] ,\delta _{ij}=\left\{ 
\begin{array}{l}
1,i=j, \\ 
0,i\neq j%
\end{array}%
\right.  \label{eq3-2b}
\end{gather}%
with known matrices $\Sigma _{\omega },\Sigma _{\nu },S.$ In this case, the
observer gain matrix can be determined using the (steady) Kalman filter
algorithm, 
\begin{gather}
L_{K}:=L=\left( APC^{T}+S\right) \Sigma _{r}^{-1},P=APA^{T}+\Sigma _{\omega
}-L_{K}\Sigma _{r}L_{K}^{T},  \label{eq3-3a} \\
\Sigma _{r}=CPC^{T}+\Sigma _{\nu }=\mathcal{E}\left(
r_{0}(k)r_{0}^{T}(k)\right) ,\mathcal{E}\left( r_{0}(i)r_{0}^{T}(j)\right)
=\Sigma _{r}\delta _{ij},  \label{eq3-3b}
\end{gather}%
the $\chi ^{2}$ test statistic is used as the evaluation function, 
\begin{equation*}
J(k)=r_{0}^{T}(k)\Sigma _{r}^{-1}r_{0}(k)\sim \mathcal{\chi }^{2}\left(
m\right) ,
\end{equation*}%
and finally the threshold $J_{th}$ is determined by means of $\mathcal{\chi }%
_{\alpha }^{2}\left( m\right) $ for a given upper-bound of false alarm rate $%
\alpha $ \cite{Ding2014}.

\subsection{Kernel attacks}

We now study the generalisation of stealthy attacks and their existence
conditions. Corresponding to the above described observer-based attack
detection strategy, we introduce the following definition.

\begin{Def}
\label{Def3-1}Given system model (\ref{eq3-1a})-(\ref{eq3-1b}) with $\omega
(k)=0,\nu (k)=0,$ and observer-based attack detector (\ref{eq2-7a})-(\ref%
{eq2-7b}), an integrity attack is stealthy if 
\begin{equation*}
\forall u,r_{0}(z)=y^{a}(z)-\hat{y}^{a}(z)=0.
\end{equation*}
\end{Def}

For our purpose, the following definition of the so-called kernel space is
introduced.

\begin{Def}
Given the plant model (\ref{eq2-1}) and a corresponding LCP $\left( \hat{M}%
(z),\hat{N}(z)\right) ,$ we call the $\mathcal{H}_{2}\times \mathcal{H}_{2}$
subspace $\mathcal{K}_{P}$ defined by 
\begin{equation}
\mathcal{K}_{P}=\left\{ \left[ 
\begin{array}{c}
u \\ 
y%
\end{array}%
\right] :\left[ 
\begin{array}{cc}
-\hat{N} & \hat{M}%
\end{array}%
\right] \left[ 
\begin{array}{c}
u \\ 
y%
\end{array}%
\right] =0,\left[ 
\begin{array}{c}
u \\ 
y%
\end{array}%
\right] \in \mathcal{H}_{2}\hspace{-2pt}\right\}  \label{eq2-9}
\end{equation}%
kernel space of the plant.
\end{Def}

It is evident that the kernel space $\mathcal{K}_{P}$ consists of all
(bounded) input and output pairs $(u,y)$ satisfying 
\begin{equation*}
\left[ 
\begin{array}{cc}
-\hat{N}(z) & \hat{M}(z)%
\end{array}%
\right] \left[ 
\begin{array}{c}
u(z) \\ 
y(z)%
\end{array}%
\right] =0.
\end{equation*}%
$\mathcal{K}_{P}$ is a closed subspace in $\mathcal{H}_{2}$ \cite%
{Vinnicombe-book}.

\bigskip

We are now in the position to present the existence condition of stealthy
attacks defined in Definition \ref{Def3-1}.

\begin{Theo}
\label{Theo3-1}Given plant model (\ref{eq3-1a})-(\ref{eq3-1b}) with $\omega
(k)=0,\nu (k)=0,$ and an observer-based attack detector (\ref{eq2-7a})-(\ref%
{eq2-7b}), an integrity attack is stealthy if and only if%
\begin{equation}
\left[ 
\begin{array}{c}
u(z) \\ 
y^{a}(z)%
\end{array}%
\right] \in \mathcal{K}_{P}.  \label{eq3-11}
\end{equation}
\end{Theo}

\begin{proof}
Without loss of generality, assume that the LCP $\left(\hat{M},\hat{N}\right)$ is
given as described in (\ref{eq2-4a}). Then, it follows from
the well-known
parameterisation of observer-based residual generators \cite%
{Ding2013}
that all observer-based residual generators (attack detectors) of
the form (%
\ref{eq2-7a})-(\ref{eq2-7b}) can be written as%
\begin{equation*}
y^{a}(z)-%
\hat{y}^{a}(z)=R(z)\left[ 
\begin{array}{cc}
-\hat{N}(z) & \hat{M}(z)%
\end{array}%
\right] \left[ 
\begin{array}{c}
u(z) \\ 
y^{a}(z)%
\end{%
array}%
\right] ,
\end{equation*}%
where $R(z)$ is a stable and \textit{%
invertible} dynamic post-filter.
Consequently, $y^{a}(z)=\hat{y}^{a}(z)$ if
and only if%
\begin{equation*}
\left[ 
\begin{array}{cc}
-\hat{N}(z) & 
\hat{M}(z)%
\end{array}%
\right] \left[ 
\begin{array}{c}
u(z) \\

y^{a}(z)%
\end{array}%
\right] =0\Longleftrightarrow \left[ 
\begin{%
array}{c}
u(z) \\ 
y^{a}(z)%
\end{array}%
\right] \in \mathcal{K}_{P}.
\end{equation*}%
The theorem is proved. 
\end{proof}

\bigskip

In this context, we introduce the definition of kernel attacks, which gives
a general form of integrity attacks that cannot be detected using an
observer detector. As will be demonstrated in the example given below, the
zero dynamics, covert and replay attacks are special forms of the kernel
attacks.

\begin{Def}
Given system model (\ref{eq3-1a})-(\ref{eq3-1b}), an integrity attack is
called kernel attack when condition (\ref{eq3-11}) holds.
\end{Def}

\begin{Exp}
We first check a zero dynamics attack. It is evident that for $\omega
(k)=0,\nu (k)=0,$%
\begin{equation*}
r_{0}(z)=y^{a}(z)-\hat{y}^{a}(z)=\left[ 
\begin{array}{cc}
-\hat{N}(z) & \hat{M}(z)%
\end{array}%
\right] \left[ 
\begin{array}{c}
u(z) \\ 
y^{a}(z)%
\end{array}%
\right] =-\hat{N}(z)a_{u}(z).
\end{equation*}%
It follows from (\ref{eq2-11b}) that 
\begin{equation*}
\hat{N}(z)a_{u}(z)=0\Longrightarrow \left[ 
\begin{array}{cc}
-\hat{N}(z) & \hat{M}(z)%
\end{array}%
\right] \left[ 
\begin{array}{c}
u(z) \\ 
y^{a}(z)%
\end{array}%
\right] =0,
\end{equation*}%
i.e. the zero dynamics attack is a kernel attack.

\bigskip

Now, consider the residual dynamics under a covert attack, which is given by%
\begin{equation*}
r_{0}(z)=\left[ 
\begin{array}{cc}
-\hat{N}(z) & \hat{M}(z)%
\end{array}%
\right] \left[ 
\begin{array}{c}
u(z) \\ 
y^{a}(z)%
\end{array}%
\right] =\hat{N}(z)a_{u}(z)+\hat{M}(z)a_{y}(z).
\end{equation*}%
According to (\ref{eq2-17a}), it holds 
\begin{equation*}
\hat{N}(z)a_{u}(z)+\hat{M}(z)a_{y}(z)=0\Longrightarrow \left[ 
\begin{array}{cc}
-\hat{N}(z) & \hat{M}(z)%
\end{array}%
\right] \left[ 
\begin{array}{c}
u(z) \\ 
y^{a}(z)%
\end{array}%
\right] =0.
\end{equation*}%
Therefore, the covert attack is obviously a kernel attack.

\bigskip

Concerning the residual dynamics under a replay attack, recall the relation (%
\ref{eq2-20}). It turns out, for $j\in \left[ k,k+M\right] ,$ 
\begin{align}
r_{0}(j)& =y^{a}(j)-\hat{y}^{a}(j)\approx Cx(j)+Du(j)-\left( C\hat{x}%
(j)+Du(j)\right) =0 \\
& \Longrightarrow \left[ 
\begin{array}{c}
u \\ 
y^{a}%
\end{array}%
\right] \in \mathcal{K}_{P}.
\end{align}%
Thus, the replay attack is a kernel attack.
\end{Exp}

Given the plant dynamics described by (\ref{eq3-1a})-(\ref{eq3-1b}) with $%
\omega (k)=0,\nu (k)=0,$ the dynamics of the observer-based attack detector (%
\ref{eq2-7a})-(\ref{eq2-7b}) is described by%
\begin{equation*}
r_{0}(z)=y^{a}(z)-\hat{y}^{a}(z)=\hat{M}(z)a_{y}(z)+\hat{N}(z)a_{u}(z).
\end{equation*}%
Consequently, if the attack pair $\left( a_{u},a_{y}\right) $ is constructed
satisfying 
\begin{equation}
\hat{M}(z)a_{y}(z)+\hat{N}(z)a_{u}(z)=0,  \label{eq3-13}
\end{equation}%
it cannot be detected. Hence, we have the following theorem.

\begin{Theo}
Given the plant model (\ref{eq3-1a})-(\ref{eq3-1b}) and an observer-based
attack detector (\ref{eq2-7a})-(\ref{eq2-7b}), the pair $\left(
a_{u},a_{y}\right) $ builds a kernel attack if it satisfies (\ref{eq3-13}).
\end{Theo}

\begin{Rem}
We would like to point out that a replay attack does not satisfy (\ref{eq3-13}),
since it is, in fact, not an additive type of attacks, as remarked in the
previous section.
\end{Rem}

\bigskip

Recall that the kernel space $\mathcal{K}_{P}$ is a structural property of
the plant and determined by the dynamics of the nominal plant. As a result,
if an attacker is in possession of knowledge of the plant dynamics, kernel
attacks could be constructed according to (\ref{eq3-13}) and injected in the
plant without being detected by the observer-based attack detector (\ref%
{eq2-7a})-(\ref{eq2-7b}). In fact, recall that any LTI observer-based
residual generators, including the parity relation based one and diagnosis
observer, can be parameterised by \cite{Ding2013} 
\begin{equation}
r(z)=R(z)r_{0}(z)=R(z)\left( y(z)-\hat{y}(z)\right) ,R(z)\neq 0,R(z)\in 
\mathcal{RH}_{\infty }.  \label{eq3-30}
\end{equation}%
The following corollary is obvious.

\begin{Corol}
\label{Coro3-1}Given plant model (\ref{eq3-1a})-(\ref{eq3-1b}), any attack
cannot be detected by an LTI attack detector of the form (\ref{eq3-30}), if
and only if the signal pair $\left( u,y^{a}\right) $ satisfies (\ref{eq3-11}%
), or if the attack signal pair $\left( a_{u},a_{y}\right) $ is constructed
satisfying (\ref{eq3-13}).
\end{Corol}

At the end of this section, we would like to emphasise that the concept of
the kernel attacks and the associated existence conditions given in Theorems
1 - 2 and Corollary \ref{Coro3-1} are described in terms of the LCF or
kernel space of the system under consideration. They are the system
structural properties and independent of the observer design and the
evaluation schemes adopted by the attack-detector. Our subsequent
investigation on detecting stealthy integrity attacks will be carried out in
this context.

\section{Analysis and detection of kernel attacks}

This section deals with detecting kernel attacks on the feedback control
systems shown in Figure 2. Consider that the observer-based attack detector (%
\ref{eq2-7a})-(\ref{eq2-7b}) performs the online detection by means of the
(online) data $(u(k),y^{a}(k)),$ whose dimension is $p+m.$ On the other
hand, the parameterisation of observer-based residual generators and Theorem %
\ref{Theo3-1} as well as Corollary \ref{Coro3-1} reveal that the residual
signals that belong to the $m$-dimensional kernel space $\mathcal{K}_{P}$
are only effective in detecting attacks which do not belong to $\mathcal{K}%
_{P}.$ In other words, in order to detect kernel attacks successfully,
generating additional signals to cover the overall $\left( p+m\right) $%
-dimensional data space is an effective alternative solution. In fact, the
so-called moving target or auxiliary system schemes reported, for instance,
in \cite%
{MT-method-CDC2015,MT-method-IEEE-TAC2020,Zhang-CDC2017,DIBAJI2019-survey}
for detecting zero dynamics and covert attacks, are special realisations of
this idea.

\begin{figure}[h]
\centering\includegraphics[width=7cm,height=3.5cm]{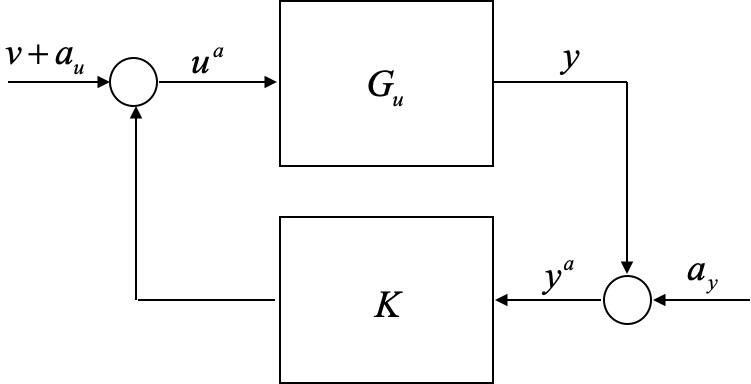}
\caption{Schematic description of the control loop under attack}
\end{figure}

\subsection{Analysis of closed-loop dynamics under kernel attacks}

Consider the feedback control loop with the plant model (\ref{eq3-1a})-(\ref%
{eq3-1b}) and controller described by (\ref{eq2-12c}), where $K(z)$
satisfies (\ref{eq2-12a})-(\ref{eq2-12b}). For the sake of the structural
analysis, $\omega (k)$ and $\nu (k)$ are assumed to be zero. It yields,
under attacks,%
\begin{align*}
y^{a}(z)& =G_{u}(z)u^{a}(z)+a_{y}(z), \\
u^{a}(z)& =K(z)y^{a}(z)+v(z)+a_{u}(z),
\end{align*}%
and the control loop configuration can be equivalently sketched by Figure 2.
It turns out 
\begin{gather}
\left[ 
\begin{array}{c}
u^{a}(z) \\ 
y^{a}(z)%
\end{array}%
\right] =\left[ 
\begin{array}{cc}
I & -K(z) \\ 
-G_{u}(z) & I%
\end{array}%
\right] ^{-1}\left[ 
\begin{array}{c}
a_{u}(z)+v(z) \\ 
a_{y}(z)%
\end{array}%
\right]  \label{eq4-1a} \\
=\left[ 
\begin{array}{cc}
\hat{V}(z) & -\hat{U}(z) \\ 
-\hat{N}(z) & \hat{M}(z)%
\end{array}%
\right] ^{-1}\left[ 
\begin{array}{c}
\hat{V}(z)\left( a_{u}(z)+v(z)\right) \\ 
\hat{M}(z)a_{y}(z)%
\end{array}%
\right] ,  \label{eq4-1b}
\end{gather}%
where 
\begin{equation*}
\hat{V}=X-Q\hat{N}\in \mathcal{RH}_{\infty },\hat{U}=-Y-Q\hat{M}\in \mathcal{%
RH}_{\infty },
\end{equation*}%
and $\hat{V},\hat{M}$ are invertible. Recall the Bezout identity (\ref{eq2-5}%
) and extend it to%
\begin{gather}
\left[ 
\begin{array}{cc}
X-Q\hat{N} & \text{ }Y+Q\hat{M} \\ 
-\hat{N} & \text{ }\hat{M}%
\end{array}%
\right] \left[ 
\begin{array}{cc}
M & \text{ }-\hat{Y}-MQ \\ 
N & \text{ }\hat{X}-NQ%
\end{array}%
\right] =\left[ 
\begin{array}{cc}
I\text{ } & 0\text{ } \\ 
0\text{ } & I\text{ }%
\end{array}%
\right] \Longleftrightarrow  \label{eq2-5a} \\
\left[ 
\begin{array}{cc}
X-Q\hat{N} & \text{ }Y+Q\hat{M} \\ 
-\hat{N} & \text{ }\hat{M}%
\end{array}%
\right] ^{-1}=\left[ 
\begin{array}{cc}
M & \text{ }-\hat{Y}-MQ \\ 
N & \text{ }\hat{X}-NQ%
\end{array}%
\right] \in \mathcal{RH}_{\infty }.  \label{eq2-5b}
\end{gather}%
As a result, we have

\begin{Theo}
\label{Theo4-1}Given the plant model (\ref{eq3-1a})-(\ref{eq3-1b}) and
controller (\ref{eq2-12c}) with $K(z)$ satisfying (\ref{eq2-12a})-(\ref%
{eq2-12b}), it holds%
\begin{gather}
\left[ 
\begin{array}{c}
\bar{a}_{u} \\ 
\bar{a}_{y}%
\end{array}%
\right] =\left[ 
\begin{array}{cc}
X-Q\hat{N} & \text{ }Y+Q\hat{M} \\ 
-\hat{N} & \text{ }\hat{M}%
\end{array}%
\right] \left[ 
\begin{array}{c}
u^{a} \\ 
y^{a}%
\end{array}%
\right] -\left[ 
\begin{array}{c}
\bar{v} \\ 
0%
\end{array}%
\right] ,\text{ }  \label{eq3-12} \\
\bar{v}=\hat{V}v,\bar{a}_{u}=\hat{V}a_{u},\bar{a}_{y}=\hat{M}a_{y}.  \notag
\end{gather}
\end{Theo}

It follows immediately from Theorem \ref{Theo4-1} that, using signals $%
y^{a}(k),u^{a}(k)$ and $v(k),$

\begin{itemize}
\item the attack signals $a_{y}(k),a_{u}(k)$ could be structurally detected
in the sense that 
\begin{gather*}
\left[ 
\begin{array}{c}
a_{u}(z) \\ 
a_{y}(z)%
\end{array}%
\right] \neq 0\Longleftrightarrow \left[ 
\begin{array}{c}
\bar{a}_{u}(z) \\ 
\bar{a}_{y}(z)%
\end{array}%
\right] \neq 0\Longleftrightarrow \\
\left[ 
\begin{array}{cc}
X-Q\hat{N} & \text{ }Y+Q\hat{M} \\ 
-\hat{N} & \text{ }\hat{M}%
\end{array}%
\right] \left[ 
\begin{array}{c}
u^{a} \\ 
y^{a}%
\end{array}%
\right] -\left[ 
\begin{array}{c}
\bar{v} \\ 
0%
\end{array}%
\right] \neq 0,\text{ and}
\end{gather*}

\item if $\hat{V}^{-1}\in \mathcal{RH}_{\infty },\hat{M}^{-1}\in \mathcal{RH}%
_{\infty },$ i.e. both the plant and controller are stable, the attack pair $%
\left( a_{y},a_{u}\right) $ could also be (structurally) uniquely identified
according to%
\begin{equation*}
\left[ 
\begin{array}{c}
a_{u} \\ 
a_{y}%
\end{array}%
\right] =\left[ 
\begin{array}{cc}
I & -K \\ 
-G_{u} & \text{ }I%
\end{array}%
\right] \left[ 
\begin{array}{c}
u^{a} \\ 
y^{a}%
\end{array}%
\right] +\left[ 
\begin{array}{c}
-v \\ 
0%
\end{array}%
\right] ,
\end{equation*}%
except that the transfer matrix 
\begin{equation}
\left[ 
\begin{array}{cc}
M & \text{ }-\hat{Y}-MQ \\ 
N & \text{ }\hat{X}-NQ%
\end{array}%
\right]  \label{eq3-14}
\end{equation}%
has a transmission zero at $z=z_{0},$ i.e. 
\begin{equation}
rank\left[ 
\begin{array}{cc}
M(z_{0}) & \text{ }-\hat{Y}(z_{0})-M(z_{0})Q(z_{0}) \\ 
N(z_{0}) & \text{ }\hat{X}(z_{0})-N(z_{0})Q(z_{0})%
\end{array}%
\right] <m+p,  \label{eq3-13a}
\end{equation}%
and 
\begin{equation}
\left[ 
\begin{array}{c}
a_{u}(z) \\ 
a_{y}(z)%
\end{array}%
\right] =\left[ 
\begin{array}{c}
a_{u}(z_{0}) \\ 
a_{y}(z_{0})%
\end{array}%
\right] .  \label{eq3-13b}
\end{equation}%
Considering that transmission zeros of transfer matrix (\ref{eq3-14}) are
structural properties of the plant and the controller, and whose number is
limited, we will not address this class of possible attacks whose
realisation requires not only full knowledge of the plant and controller,
but also very special forms of attack signals.
\end{itemize}

\bigskip

It is worth noting that, according to the relations given in (\ref{eq2-5a})-(%
\ref{eq2-5b}), attacks $a_{y},a_{u}$ can also be (structurally) detected
using the relation%
\begin{gather}
\left[ 
\begin{array}{c}
u^{a} \\ 
y^{a}%
\end{array}%
\right] -\left[ 
\begin{array}{c}
M \\ 
N%
\end{array}%
\right] \bar{v}=\left[ 
\begin{array}{cc}
M & \text{ }-\hat{Y}-MQ \\ 
N & \text{ }\hat{X}-NQ%
\end{array}%
\right] \left[ 
\begin{array}{c}
\bar{a}_{u} \\ 
\bar{a}_{y}%
\end{array}%
\right]  \label{eq3-16} \\
=\left[ 
\begin{array}{cc}
I & -\hat{Y}-MQ \\ 
0 & \hat{X}-NQ%
\end{array}%
\right] \left[ 
\begin{array}{c}
a_{u} \\ 
\hat{N}a_{u}+\hat{M}a_{y}%
\end{array}%
\right] .  \notag
\end{gather}%
Before we continue our study on applying signals $v(k),y^{a}(k)$ and $%
u^{a}(k)$ for attack detection, we would like to discuss about relations (%
\ref{eq3-12}) and (\ref{eq3-16}), which is helpful to gain a deep insight
into our solutions and two different implementation forms of attack
detectors. To this end, we first check the transfer function matrix 
\begin{equation*}
\left[ 
\begin{array}{cc}
X-Q\hat{N} & \text{ }Y+Q\hat{M} \\ 
-\hat{N} & \text{ }\hat{M}%
\end{array}%
\right] =\left[ 
\begin{array}{cc}
\hat{V} & \text{ }-\hat{U} \\ 
-\hat{N} & \text{ }\hat{M}%
\end{array}%
\right]
\end{equation*}%
on the right-hand side of (\ref{eq3-12}). While the LCP $\left( \hat{M},\hat{%
N}\right) $ builds the kernel space of the plant, the pair $\left( \hat{V},%
\hat{U}\right) $ is left coprime and spans the kernel space of the
controller (\ref{eq2-12c}). In other words, the signal $r_{u}(z)$ defined by%
\begin{equation*}
u(z)-v(z)=K(z)y(z)\Longrightarrow \hat{V}(z)\left( u(z)-v(z)\right) -\hat{U}%
(z)y(z)=:r_{u}(z)
\end{equation*}%
can be viewed as a residual vector generated based on the controller
configuration. Since 
\begin{equation*}
\dim \left[ 
\begin{array}{c}
r_{0} \\ 
r_{u}%
\end{array}%
\right] =m+p,
\end{equation*}%
any changes (caused, for instance, by attacks) in the space spanned by the
plant input and output vectors, $\left( u,y\right) ,$ can be (structurally)
uniquely detected. It is of interest to notice that the residual vector $%
r_{u}(k)$ can be generated as well using an observer of the form%
\begin{align}
\hat{x}_{u}(k+1)& =\left( A-LC\right) \hat{x}_{u}(k)+(B-LD)\left(
u(k)-v(k)\right) +Ly(k),  \label{eq3-22a} \\
\left[ 
\begin{array}{c}
r_{u,1}(k) \\ 
r_{u,2}(k)%
\end{array}%
\right] & =\left[ 
\begin{array}{c}
u(k)-v(k)-F\hat{x}_{u}(k) \\ 
y(k)-D\left( u(k)-v(k)\right) -C\hat{x}_{u}(k)%
\end{array}%
\right] , \\
r_{u}(z)& =r_{u,1}(z)+Q(z)r_{u,2}(z).  \label{eq3-22b}
\end{align}%
We now consider the left-hand side of (\ref{eq3-16}). In the attack-free
case, it is indeed a residual generator based on the closed-loop dynamics, 
\begin{equation}
\left[ 
\begin{array}{c}
r_{u,c}(z) \\ 
r_{y,c}(z)%
\end{array}%
\right] :=\left[ 
\begin{array}{c}
u^{a}(z) \\ 
y^{a}(z)%
\end{array}%
\right] -\left[ 
\begin{array}{c}
M(z) \\ 
N(z)%
\end{array}%
\right] \bar{v}(z),  \label{eq3-23}
\end{equation}%
whose state space representation is given by%
\begin{align*}
\hat{x}_{v}(k+1)& =\left( A-LC\right) \hat{x}_{v}(k)+(B-LD)v(k), \\
\bar{v}(k)& =v(k)-F\hat{x}_{v}(k)-q(k),q(z)=Q(z)\left( C\hat{x}%
_{v}(z)+Dv(z)\right) , \\
\hat{x}_{c}(k+1)& =\left( A+BF\right) \hat{x}_{c}(k)+B\bar{v}(k), \\
\left[ 
\begin{array}{c}
r_{u,c}(k) \\ 
r_{y,c}(k)%
\end{array}%
\right] & =\left[ 
\begin{array}{c}
u^{a}(k) \\ 
y^{a}(k)%
\end{array}%
\right] -\left[ 
\begin{array}{c}
F \\ 
C+DF%
\end{array}%
\right] \hat{x}_{c}(k)-\left[ 
\begin{array}{c}
I \\ 
D%
\end{array}%
\right] \bar{v}(k).
\end{align*}%
Hence, using the closed-loop dynamics based residual vectors, $r_{u,c}$ and $%
r_{y,c},$ with 
\begin{equation*}
\dim \left[ 
\begin{array}{c}
r_{u,c} \\ 
r_{y,c}%
\end{array}%
\right] =m+p,
\end{equation*}%
it is possible to detect attacks uniquely as well.

\bigskip

In summary, in order to detect all kernel attacks uniquely, we can use the
(online) data $v(k),y^{a}(k)$ and $u^{a}(k)$ to generate either the
observer-based residuals $r_{0}$ and $r_{u}$ or the closed-loop dynamics
based residuals $r_{u,c}$ and $r_{y,c}.$ Unfortunately, $u^{a}(k)$ is only
available on the side of the plant, as shown in Figures 1 and 2. \ This
motivates us to propose a detection scheme described in the next sub-section.

\subsection{A conceptual scheme for detecting kernel attacks}

For the realisation of the detection solution based on (\ref{eq3-12}) given
in Theorem \ref{Theo4-1}, we propose the following conceptual detection
scheme.

\bigskip

It follows from the relation 
\begin{equation*}
\left( X(z)-Q(z)\hat{N}(z)\right) \left( u(z)-v(z)\right) =-\left( Y(z)+Q(z)%
\hat{M}(z)\right) y(z)
\end{equation*}%
in attack-free case that signal%
\begin{equation}
r_{u,0}(z):=X(z)u(z)+Y(z)y(z)-\left( X(z)-Q(z)\hat{N}(z)\right) v(z)
\label{eq4-10}
\end{equation}%
builds a residual signal satisfying 
\begin{equation*}
r_{u,0}(z)=r_{u}(z)+Q(z)r_{0}(z).
\end{equation*}%
Recall that $v$ is available at the monitoring side, while the signals $y,u$
exist on the plant side with $u$ being corrupted by $a_{u}.$ For our
purpose, we propose to generate the residual signal $r_{u,0}$ using the
following algorithm:

\begin{itemize}
\item compute 
\begin{equation}
r_{en}(z):=X(z)u^{a}(z)+Y(z)y(z);  \label{eq4-9}
\end{equation}

\item transmit $r_{en}(k)$ to the monitoring and control side;

\item compute, on the monitoring and control side,%
\begin{equation}
r_{u,0}(z)=r_{en}^{a}(z)-\left( X(z)-Q(z)\hat{N}(z)\right) v(z).
\label{eq4-9a}
\end{equation}%
Here, it is supposed that $r_{en}(z)$ is attacked by the attack signal $%
a_{r_{en}}(k)$, i.e.%
\begin{equation*}
r_{en}^{a}(k)=r_{en}(k)+a_{r_{en}}(k).
\end{equation*}
\end{itemize}

Figure 3 shows the corresponding system configuration.

\begin{figure}[h]
\centering\includegraphics[width=11cm,height=6cm]{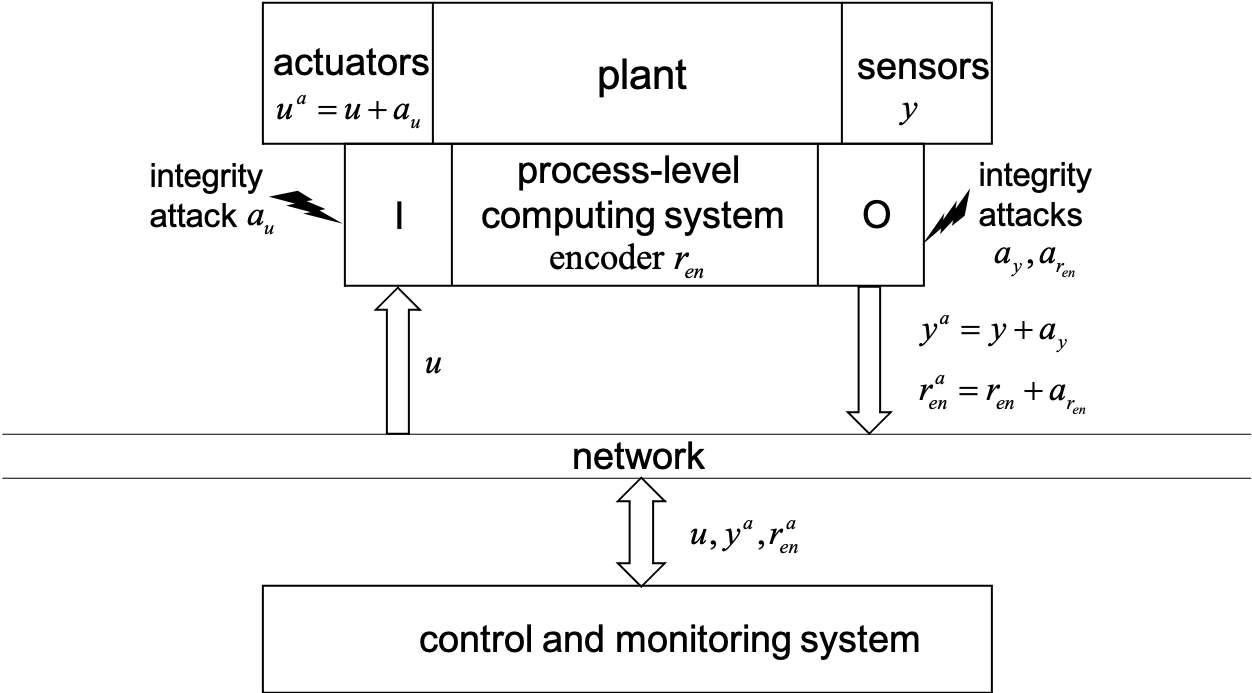}
\caption{Schematic description of kernel attack detection scheme}
\end{figure}
\bigskip
Next, we check the dynamics of $r_{u,0},r_{0}$ without considering noises.
Recall that%
\begin{gather*}
u(z)=K(z)y^{a}(z)+v(z)\Longleftrightarrow \\
\left( X(z)-Q(z)\hat{N}(z)\right) \left( u(z)-v(z)\right) =-\left( Y(z)+Q(z)%
\hat{M}(z)\right) y^{a}(z).
\end{gather*}%
It yields 
\begin{gather}
\left[ 
\begin{array}{c}
r_{u,0} \\ 
r_{0}%
\end{array}%
\right] =\left[ 
\begin{array}{ccc}
0 & 0 & I \\ 
-\hat{N} & \hat{M} & 0%
\end{array}%
\right] \left[ 
\begin{array}{c}
u \\ 
y^{a} \\ 
r_{en}^{a}%
\end{array}%
\right] -\left[ 
\begin{array}{c}
\bar{v} \\ 
0%
\end{array}%
\right]  \notag \\
=\left[ 
\begin{array}{ccc}
X-Q\hat{N} & -Y-Q\hat{M} & I \\ 
\hat{N} & \hat{M} & 0%
\end{array}%
\right] \left[ 
\begin{array}{c}
a_{u} \\ 
a_{y} \\ 
a_{r_{en}}%
\end{array}%
\right] .  \label{eq4-6}
\end{gather}%
Consequently, 
\begin{equation*}
\left[ 
\begin{array}{c}
r_{u,0}(z) \\ 
r_{0}(z)%
\end{array}%
\right] =0
\end{equation*}%
if and only if $a_{r_{en}},a_{u},a_{y}$ solve 
\begin{gather}
\hat{N}(z)a_{u}(z)+\hat{M}(z)a_{y}(z)=0,  \label{eq4-7} \\
a_{r_{en}}(z)=Y(z)a_{y}(z)-X(z)a_{u}(z).  \label{eq4-8}
\end{gather}%
We summary the above results in the following theorem.

\begin{Theo}
\label{Theo4-2}Given the plant model (\ref{eq3-1a})-(\ref{eq3-1b}), the
controller%
\begin{equation*}
u(z)=K(z)y^{a}(z)+v(z),
\end{equation*}%
with $K(z)$ satisfying (\ref{eq2-12a})-(\ref{eq2-12b}) and the system
configuration shown in Figure 3, where $a_{r_{en}},$ $a_{u},$ $a_{y}$ are
the attack signals, $r_{en}^{a},y^{a},u,v$ are the system signals being
available at the monitoring side and used for the attack detection purpose,
then attacks $\left( a_{r_{en}},a_{u},a_{y}\right) $ are stealthy, i.e. they
cannot be detected using the available system signals, if and only if the
conditions (\ref{eq4-7})-(\ref{eq4-8}) are satisfied.
\end{Theo}

\subsection{\textbf{Design and construction of }residual generator $r_{u,0}$%
\label{Sec4-3}}

It follows from Theorem \ref{Theo4-2} that an attacker could design attack
signals $a_{r_{en}},a_{u},$ $a_{y}$ so that conditions (\ref{eq4-7})-(\ref%
{eq4-8}) are satisfied when the attacker is in possession of knowledge of
the plant dynamics (regarding to (\ref{eq4-7})) and the construction of
system (\ref{eq4-9}) (regarding to (\ref{eq4-8})). This demands that
knowledge of the residual generator (\ref{eq4-9}) and (\ref{eq4-10}) should
be protected from the attacker so that the attacker could not be able to
construct the attack signal $a_{r_{en}}$ according to (\ref{eq4-8}). To this
end, encryption of the concerning system dynamics is the major task of
designing and constructing the residual generator $r_{u,0}$ described by (%
\ref{eq4-9}) and (\ref{eq4-10}). This will be realised in two steps.

\bigskip

At first, the residual generator (\ref{eq4-10}) is constructed at two
different sides of the networked control system, as shown in (\ref{eq4-9})
and (\ref{eq4-9a}). In a certain sense, the dynamic system 
\begin{equation*}
r_{en}(z)=\left[ 
\begin{array}{cc}
X(z) & \text{ }Y(z)%
\end{array}%
\right] \left[ 
\begin{array}{c}
u^{a}(z) \\ 
y(z)%
\end{array}%
\right]
\end{equation*}%
can be interpreted as an encoding algorithm and thus is called encoder. The
residual signal $r_{u,0}$ is then generated by a decoding algorithm in the
form 
\begin{equation*}
r_{u,0}(z)=r_{en}(z)-\left( X(z)-Q(z)\hat{N}(z)\right) v(z).
\end{equation*}%
Note that even if knowledge of the (encoding) system (\ref{eq4-9}) is
protected, the attacker could identify the system dynamics using possibly
eavesdropped signals $r_{en},u^{a}$ and $y.$ In order to protect the system
dynamics from being identified, the involved system $(X, Y)$ is further
encrypted in the next step.

\bigskip

Recall that the state space form of the encoder (\ref{eq4-9}) is given by 
\begin{align}
\varsigma (k+1)& =\left( A-LC\right) \varsigma (k)+Ly(k)+(B-LD)u^{a}(k), \\
r_{en}(k)& =u^{a}(k)-F\varsigma (k).
\end{align}%
Moreover, the following lemma can be proved.

\begin{Le}
\label{Le4-1}Given $\left( \hat{M}_{i},\hat{N}_{i}\right) ,\left(
X_{i},Y_{i}\right) ,i=1,2,$ subject to 
\begin{align*}
\hat{M}_{i}& =\left( A-L_{i}C,-L_{i},C,I\right) ,\hat{N}_{i}=\left(
A-L_{i}C,B-L_{i}D,C,D\right) , \\
X_{i}& =\left( A-L_{i}C,-(B-L_{i}D),F_{i},I\right) ,Y_{i}=\left(
A-L_{i}C,-L_{i},F_{i},0\right) ,
\end{align*}%
it holds%
\begin{align}
\left[ 
\begin{array}{cc}
X_{1}(z) & Y_{1}(z)%
\end{array}%
\right] & =R_{12}(z)\left[ 
\begin{array}{cc}
X_{2}(z) & Y_{2}(z)%
\end{array}%
\right] +\bar{Q}_{11}(z)\left[ 
\begin{array}{cc}
-\hat{N}_{1}(z) & \hat{M}_{1}(z)%
\end{array}%
\right]  \label{eq4-20} \\
& =R_{12}(z)\left[ 
\begin{array}{cc}
X_{2}(z) & Y_{2}(z)%
\end{array}%
\right] +\bar{Q}_{12}(z)\left[ 
\begin{array}{cc}
-\hat{N}_{2}(z) & \hat{M}_{2}(z)%
\end{array}%
\right] ,  \label{eq4-20a} \\
R_{12}(z)& =R_{21}^{-1}(z)=I+\left( F_{2}-F_{1}\right) \left(
zI-A_{F_{2}}\right) ^{-1}B\in \mathcal{RH}_{\infty },A_{F_{i}}=A+BF_{i}, 
\notag \\
\bar{Q}_{11}(z)& =F_{1}\left( zI-A_{L_{2}}\right) ^{-1}\left(
L_{2}-L_{1}\right) -\bar{R}_{12}(z)Q_{21}(z)\in \mathcal{RH}_{\infty
},A_{L_{i}}=A-L_{i}C,  \notag \\
\bar{Q}_{12}(z)& =F_{1}\left( zI-A_{L_{1}}\right) ^{-1}\left(
L_{2}-L_{1}\right) -\left( F_{1}-F_{2}\right) \left( zI-A_{F_{2}}\right)
^{-1}L_{2}\in \mathcal{RH}_{\infty },  \notag \\
\bar{R}_{12}(z)& =\left( F_{1}-F_{2}\right) \left( zI-A_{F_{2}}\right)
^{-1}L_{2}\in \mathcal{RH}_{\infty }, \\
Q_{21}(z)& =Q_{12}^{-1}(z)=I+C\left( zI-A_{L_{2}}\right) ^{-1}\left(
L_{1}-L_{2}\right) \in \mathcal{RH}_{\infty }.
\end{align}
\end{Le}

The proof is given in Appendix.

\bigskip

Lemma \ref{Le4-1} reveals that varying the gain matrices $F_{2}$ and $L_{2}$
in $\left( X,Y\right) $ to $F_{1}$ and $L_{1}$ is equivalent to adding (i) a
(stable and invertible) post-filter to $r_{en}(z)$ and (ii) additional
residual signal $r_{0}$. On the basis of this result, we propose to switch $%
F $ and $L,$ denoted by $F_{\sigma }$ and $L_{\sigma },$ among a set of
values as follows:%
\begin{align*}
F_{\sigma }& \in \mathcal{F}:=\left\{ F_{i}\in \mathcal{R}^{p\times
n},A+BF_{i}\text{ is Schur},i\in \mathcal{I}\right\} ,\mathcal{I}=\left\{
1,\ldots ,\kappa \right\} , \\
L_{\sigma }& \in \mathcal{L}:=\left\{ L_{i}\in \mathcal{R}^{n\times
m},A-L_{i}C\text{ is Schur},i\in \mathcal{I}\right\} ,
\end{align*}%
where $\sigma \in \mathcal{I}$ is the switching law that is to be protected
so that it is unknown for the attacker. Let $F_{0}$ and $L_{0}$ denote the
gain matrices $F$ and $L$ adopted in the control law (\ref{eq2-12a})-(\ref%
{eq2-12b}), and 
\begin{align*}
r_{0,\sigma }(z)& =\hat{M}_{\sigma }(z)y(z)-\hat{N}_{\sigma
}(z)u^{a}(z),\sigma \in \mathcal{I}, \\
r_{en,\sigma }(z)& =X_{\sigma }(z)u^{a}(z)+Y_{\sigma }(z)y(z),
\end{align*}%
where%
\begin{align*}
\hat{M}_{\sigma }& =\left( A-L_{\sigma }C,-L_{\sigma },C,I\right) ,\hat{N}%
_{\sigma }=\left( A-L_{\sigma }C,B-L_{\sigma }D,C,D\right) , \\
X_{\sigma }& =\left( A-L_{\sigma }C,-(B-L_{\sigma }D),F_{\sigma },I\right)
,Y_{\sigma }=\left( A-L_{\sigma }C,-L_{\sigma },F_{\sigma },0\right) .
\end{align*}%
Then, we have

\begin{Theo}
\label{Theo4-3}Given the plant model (\ref{eq3-1a})-(\ref{eq3-1b}), the
control law $K(z)$ satisfying (\ref{eq2-12a})-(\ref{eq2-12b}) and the system
configuration shown in Figure 3, it holds%
\begin{align}
r_{0,\sigma }(z)& =P_{0,\sigma }(z)r_{0,p}(z),r_{0,p}(z)=\hat{M}_{0}(z)y(z)-%
\hat{N}_{0}(z)u^{a}(z), \\
P_{0,\sigma }(z)& =I+C\left( zI-A_{L_{\sigma }}\right) ^{-1}\left(
L_{0}-L_{\sigma }\right) \in \mathcal{RH}_{\infty },A_{L_{\sigma
}}=A-L_{\sigma }C, \\
r_{en,\sigma }(z)& =P_{u,\sigma }(z)r_{en,0}(z)+Q_{\sigma
}(z)r_{0,p}(z),r_{en,0}(z)=X_{0}(z)u^{a}(z)+Y_{0}(z)y(z), \\
P_{u,\sigma }(z)& =I+\left( F_{0}-F_{\sigma }\right) \left(
zI-A_{F_{0}}\right) ^{-1}B\in \mathcal{RH}_{\infty },A_{F_{0}}=A+BF_{0}, \\
Q_{\sigma }(z)& =F_{\sigma }\left( zI-A_{L_{\sigma }}\right) ^{-1}\left(
L_{0}-L_{\sigma }\right) -\left( F_{\sigma }-F_{0}\right) \left(
zI-A_{F_{0}}\right) ^{-1}L_{0}\in \mathcal{RH}_{\infty },
\end{align}%
where%
\begin{align*}
\hat{M}_{0}& =\left( A-L_{0}C,-L_{0},C,I\right) ,\hat{N}_{0}=\left(
A-L_{0}C,B-L_{0}D,C,D\right) , \\
X_{0}& =\left( A-L_{0}C,-(B-L_{0}D),F_{0},I\right) ,Y_{0}=\left(
A-L_{0}C,-L_{0},F_{0},0\right) .
\end{align*}
\end{Theo}

The proof of this theorem follows immediately from Lemma \ref{Le4-1}.

\begin{Rem}
Although in the attack-free case 
\begin{equation*}
r_{0}(z)=r_{0,p}(z)=\hat{M}_{0}(z)y(z)-\hat{N}_{0}(z)u(z)=r_{0,n}(z),
\end{equation*}%
we would like to call the reader's attention to the differences between the
residual signals $r_{0,p}$ and $r_{0}.$ While $r_{0}$ is realised on the
monitoring and control side, $r_{0,p}$ is generated on the plant side. Here, $r_{0,n}$ denotes the influence of the noises on the residual vector.
Moreover, in case of attacks, 
\begin{align*}
r_{0,p}(z)& =\hat{M}_{0}(z)y(z)-\hat{N}_{0}(z)u^{a}(z)+r_{0,n}(z)=r_{0,n}(z),
\\
r_{0}(z)& =\hat{M}_{0}(z)y^{a}(z)-\hat{N}_{0}(z)u(z)=\hat{N}_{0}(z)a_{u}(z)+%
\hat{M}_{0}(z)a_{y}(z)+r_{0,n}(z).
\end{align*}%
\end{Rem}

Theorem \ref{Theo4-3} demonstrates that switching the gain matrices $\left(
F_{\sigma },L_{\sigma }\right) $ can be equivalently interpreted as (i)
switching post-filters $P_{u,\sigma }(z)$ and $P_{0,\sigma }(z)$ to $%
r_{en,0}(z)$ and $r_{0,p}(z),$ 
\begin{align*}
P_{u,\sigma }(z)& \in \left\{ P_{u,i}(z)\in \mathcal{RH}_{\infty
},P_{u,i}=I+\left( F_{0}-F_{i}\right) \left( zI-A_{F_{0}}\right) ^{-1}B,i\in 
\mathcal{I}\right\} , \\
P_{0,\sigma }(z)& \in \left\{ P_{0,i}(z)\in \mathcal{RH}_{\infty
},P_{0,i}=I+C\left( zI-A_{L_{i}}\right) ^{-1}\left( L_{0}-L_{i}\right) ,i\in 
\mathcal{I}\right\} ,
\end{align*}%
and (ii) adding additional residual signal $Q_{\sigma }(z)r_{0,p}(z)$ with a
switching post-filter $Q_{\sigma }(z),$%
\begin{equation}
Q_{\sigma }(z)\in \left\{ 
\begin{array}{c}
Q_{i}(z)\in \mathcal{RH}_{\infty },i\in \mathcal{I}, \\ 
Q_{i}=F_{0}\left( zI-A_{L_{i}}\right) ^{-1}\left( L_{0}-L_{i}\right)
P_{0,i}(z)-\left( F_{i}-F_{0}\right) \left( zI-A_{F_{0}}\right) ^{-1}L_{0}%
\end{array}%
\right\} .
\end{equation}%
Note that $Q_{\sigma }(z)r_{0,p}(z)$ is noise. In the remaining part of this
work, $F_{0}$ and $L_{0}$ are used to denote the gain matrices $F$ and $L$
adopted in the control law (\ref{eq2-12a})-(\ref{eq2-12b}), and
correspondingly the LCP $\left( \hat{M},\hat{N}\right) ,\left( X,Y\right) $
are denoted by $\left( \hat{M}_{0},\hat{N}_{0}\right) ,\left(
X_{0},Y_{0}\right) ,$ respectively.

\bigskip

\begin{Rem}
The encrypting effect of adding switched post filters and noises by
switching the gain matrices among different values is analogues to the
existing approaches, for instance, reported in \cite%
{MT-method-CDC2015,MT-method-IEEE-TAC2020,Zhang-CDC2017,DIBAJI2019-survey},
although in our proposed method both the design and (online) computations
are considerably less demanding. In addition, the signal $r_{en}$ is encoded
on the plant side before the transmission and the real residual signal $%
r_{u,0}$ is recovered by a decoding algorithm on the monitoring and control
side.
\end{Rem}

The switched encoder system plays a central role for detecting the kernel
attacks successfully. Their use is to prevent an attacker from identifying
the dynamics of encoding system (\ref{eq4-9}) so that the attack signal $%
a_{r_{en}}(k)$ is set to be%
\begin{equation*}
a_{r_{en}}(z)=Y_{\sigma }(z)a_{y}(z)-X_{\sigma }(z)a_{u}(z).
\end{equation*}%
On the other hand, the needed online computations for the implementation of 
\begin{equation}
r_{en,\sigma }(z)=X_{\sigma }(z)u^{a}(z)+Y_{\sigma }(z)y(z)  \label{eq4-30}
\end{equation}%
that is to be performed on the plant side should be considered and kept as
less as possible.

\bigskip

Let $\sigma \left( k_{s}\right) $ denote the switching law with $k_{s}$ as
switching time instant, $F_{\sigma \left( k_{s}\right) }$ and $L_{\sigma
\left( k_{s}\right) }$ be the operating mode of the gain matrices between
two successive switching time instants $k_{s}=k_{0},k_{1}$. On the
assumption that

\begin{itemize}
\item the attacker could access $y(k),u^{a}(k)$, even

\item have knowledge of $F_{i}$ and $L_{i}$ and so that $\left(
X_{i},Y_{i}\right) ,i=1,\cdots ,\kappa ,$ are known,

\item the switching law $\sigma \left( k_{s}\right) $ is shared only by the
monitoring system and the plant system but kept hidden from the attacker,
\end{itemize}

the LCP $\left( X_{\sigma \left( k_{0}\right) },Y_{\sigma \left(
k_{0}\right) }\right) $ (i.e. the encoder (\ref{eq4-30}) running over the
time interval $[k_{0},k_{1})$) should not be detected or identified by the
attacker using the data collected over $[k_{0},k_{1}).$ This can be
formulated as an inverse problem of fault isolation or identification. It is
well-known that if the time interval $[k_{0},k_{1})$ is sufficiently short
with respect to the complexity (e.g. the order) of $\left( X_{\sigma \left(
k_{0}\right) },Y_{\sigma \left( k_{0}\right) }\right) $ and the mode number $%
\kappa ,$ with high confidential $\left( X_{\sigma \left( k_{0}\right)
},Y_{\sigma \left( k_{0}\right) }\right) $ cannot be detected or identified.
On the other hand, in order to guarantee the stability of the switched
system, the switching law $\sigma \left( k_{s}\right) $ is to be designed to
satisfy the so-called average dwell time (ADT) condition \cite%
{HM-CDC1999,ZZSL-IEEE-TAC2012}. Recall that $r_{en,\sigma }$ is only used
for the detection purpose and $\left( X_{\sigma \left( k_{0}\right)
},Y_{\sigma \left( k_{0}\right) }\right) $ has, different from the existing
approaches, no influence on the system control performance. As a result, $%
\left( X_{\sigma \left( k_{0}\right) },Y_{\sigma \left( k_{0}\right)
}\right) $ together with the switching law $\sigma \left( k_{s}\right) $ can
be designed so that (i) $\left( X_{\sigma \left( k_{0}\right) },Y_{\sigma
\left( k_{0}\right) }\right) $ is not identifiable over the time interval,
(ii) the ADT condition is satisfied. Since the major focus of this work is
on detecting kernel attacks, we will not discuss about the design of the
switching issues for $F_{\sigma }$ and $L_{\sigma }$ in more details. The
reader can refer to, for instance, the approach of cryptographically secure
pseudo random number generator (PRNG) described in \cite%
{MT-method-IEEE-TAC2020} or the approach proposed by \cite{Zhang-CDC2017}.

\subsection{Realisation of the detection scheme\label{sub-sec4-4}}

In this sub-section, we describe the realisation of the detection scheme
proposed in the previous sub-section. To this end, two issues are to be
addressed: (i) real-time implementation of the residual generators, and (ii)
design of test statistic and threshold setting. Concerning the first issue,
the major tasks consist of

\begin{itemize}
\item computation on the plant side:%
\begin{equation}
r_{en,\sigma }(z)=X_{\sigma }(z)u^{a}(z)+Y_{\sigma }(z)y(z),  \label{eq4-21}
\end{equation}
\item signal transmissions from the plant side to the monitoring side:%
\begin{equation*}
r_{en,\sigma }^{a}(k)=r_{en,\sigma }(k)+a_{r_{en}}(k),y^{a}=y(k)+a_{y}(k),
\end{equation*}
\item computation on the monitoring and control side:%
\begin{align}
r_{u,0}(z)& =r_{en,\sigma }^{a}(z)-P_{u,\sigma }(z)\bar{v}_{0}(z),\bar{v}%
_{0}(z)=\left( X_{0}(z)-Q(z)\hat{N}_{0}(z)\right) v(z),  \label{eq4-22} \\
r_{0}(z)& =\hat{M}_{0}(z)y^{a}(z)-\hat{N}_{0}(z)u(z),  \label{eq4-23}
\end{align}
\end{itemize}

and under consideration of the plant model (\ref{eq3-1a})-(\ref{eq3-1b})
with the process and sensor noises satisfying (\ref{eq3-2a})-(\ref{eq3-2b}).
It follows from (\ref{eq2-4a})-(\ref{eq2-6b}) that the state space
realisations of (\ref{eq4-21})-(\ref{eq4-23}) are described respectively by%
\begin{align}
\varsigma (k+1)& =\left( A-L_{\sigma }C\right) \varsigma (k)+L_{\sigma
}y(k)+(B-L_{\sigma }D)u^{a}(k), \\
r_{en,\sigma }(k)& =u^{a}(k)-F_{\sigma }\varsigma (k)
\end{align}%
as well as%
\begin{align}
\hat{x}(k+1)& =\left( A-L_{0}C\right) \hat{x}%
(k)+(B-L_{0}D)u(k)+L_{0}y^{a}(k),  \label{eq3-19a} \\
x_{v}(k+1)& =\left( A-L_{0}C\right) x_{v}(k)+(B-L_{0}D)v(k),  \label{eq3-19b}
\\
r_{0}(k)& =y^{a}(k)-\left( C\hat{x}(k)+Du(k)\right) , \\
\bar{v}_{0}(z)& =v(z)-Fx_{v}(z)-Q(z)\left( Cx_{v}(z)+Dv(z)\right) ,
\label{eq4-19c} \\
r_{u,0}(z)& =r_{en,\sigma }^{a}(z)-P_{u,\sigma }(z)\bar{v}_{0}(z).
\label{eq4-19d}
\end{align}%
Next, the influences of $\omega (k),\nu (k)$ on $r_{u,0}(k)$ and $r_{0}(k)$
during attack-free operations are analysed aiming at setting an optimal
threshold. It turns out%
\begin{gather}
e(k+1)=\left( A-L_{0}C\right) e(k)+\omega (k)-L_{0}\nu (k),e(k)=x(k)-\hat{x}%
(k),  \label{eq4-2a} \\
r_{0}(k)=Ce(k)+\nu (k), \\
r_{u,0}(z)=r_{en,\sigma }^{a}(z)-P_{u,\sigma }(z)\bar{v}_{0}(z)=P_{u,\sigma
}(z)\left( r_{en,0}(z)-\bar{v}_{0}(z)\right) +Q_{\sigma }(z)r_{0}(z), \\
r_{en,0}(z)-\bar{v}_{0}(z)=-Q(z)r_{0}(z)\Longrightarrow r_{u,0}(z)=\left(
Q_{\sigma }(z)-P_{u,\sigma }(z)Q(z)\right) r_{0}(z),  \label{eq4-2b}
\end{gather}%
which implies that the residual vector%
\begin{equation*}
\left[ 
\begin{array}{c}
r_{u,0}(z) \\ 
r_{0}(z)%
\end{array}%
\right] =\left[ 
\begin{array}{c}
\bar{Q}_{\sigma }(z) \\ 
I%
\end{array}%
\right] r_{0}(z),\bar{Q}_{\sigma }=Q_{\sigma }-P_{u,\sigma }Q
\end{equation*}%
is a normally distributed color noise vector. In order to achieve an optimal
attack detection, a post-filter $P(z)$ is added as follows%
\begin{gather}
r(z)=\left[ 
\begin{array}{c}
r_{u}(z) \\ 
r_{0,K}(z)%
\end{array}%
\right] :=P(z)\left[ 
\begin{array}{c}
r_{u,0}(z) \\ 
r_{0}(z)%
\end{array}%
\right] ,  \label{eq4-24} \\
P(z)=\left[ 
\begin{array}{cc}
I & -\bar{Q}_{\sigma }(z) \\ 
0 & Q_{K0}(z)%
\end{array}%
\right] ,Q_{K0}(z)=I+C\left( zI-A_{L_{K}}\right) ^{-1}\left(
L_{0}-L_{K}\right)  \label{eq4-24a} \\
\Longrightarrow r_{u}(z):=r_{u,0}(z)-\bar{Q}_{\sigma }(z)r_{0}(z)=0,  \notag
\\
e(k+1)=A_{L_{K}}e(k)+\omega (k)-L_{K}\nu (k),A_{L_{K}}=A-L_{K}C,  \notag \\
r_{0,K}(k)=r_{K}(k)=Ce(k)+\nu (k),
\end{gather}%
where $L_{K}$ is the Kalman filter gain matrix satisfying (\ref{eq3-3a}).
Correspondingly, $r_{K}(k)\sim \mathcal{N}\left( 0,\Sigma _{r}\right) $ and
is white with $\Sigma _{r}$ given in (\ref{eq3-3b}). It is remarkable that
the residual vector $r_{u}$ is fully decoupled from the noises $\omega
(k),\nu (k).$ In order to define a practical and easily computing (scale)
test statistic, $r_{u}(z)$ is treated as a (quasi-) random vector with a
covariance matrix whose inverse is approximated by $\lambda I$, where $%
\lambda >0$ is a sufficiently large number. As a result, we set the test
statistic equal to 
\begin{equation}
J(k)=\lambda r_{u}^{T}(k)r_{u}(k)+r_{0,K}^{T}(k)\Sigma
_{r}^{-1}r_{0,K}(k)\sim \mathcal{\chi }^{2}\left( m\right) ,  \label{eq4-4}
\end{equation}%
which is subject to $\mathcal{\chi }^{2}$ distribution with $m$ degrees of
freedom in the attack-free operation, and the threshold 
\begin{equation}
J_{th}=\mathcal{\chi }_{\alpha }^{2}\left( m\right)  \label{eq4-5}
\end{equation}%
for a given upper-bound of false alarm rate $\alpha $.

\begin{Rem}
It is noteworthy that detecting kernel attacks is in the foreground of our
study. In order to highlight the basic ideas and major results in this
regard clearly, only process and measurement noises are taken into account.
The above simplified handling of $r_{u}$ follows from the geometric
interpretation of the $\mathcal{\chi }^{2}$ text statistic \cite{Ding2020}.
If unknown inputs and model parameter variations are to be considered,
advanced fault detection methods could be applied \cite{Ding2020}.
\end{Rem}

When the control loop is attacked, the dynamics of the observer-based attack
detector (\ref{eq4-22})-(\ref{eq4-23}) is governed by%
\begin{gather}
r_{u,0}(z)=r_{en,\sigma }^{a}(z)-P_{u,\sigma }(z)\bar{v}_{0}(z)
\label{eq4-3a} \\
=X_{\sigma }(z)u^{a}(z)+Y_{\sigma }(z)y(z)+a_{r_{en}}(z)-P_{u,\sigma }(z)%
\bar{v}_{0}(z)  \notag \\
=P_{u,\sigma }(z)\left( X_{0}(z)u^{a}(z)+Y_{0}(z)y(z)\right) +Q_{\sigma
}(z)r_{0,p}(z)+a_{r_{en}}(z)-P_{u,\sigma }(z)\bar{v}_{0}(z)  \notag \\
=a_{1}(z)+\bar{Q}_{\sigma }(z)r_{0,n}(z),  \label{eq4-25} \\
a_{1}(z)=P_{u,\sigma }(z)\left( X_{0}(z)a_{u}(z)-Y_{0}(z)a_{y}(z)\right)
+a_{r_{en}}(z),  \notag \\
r_{0}(z)=\hat{M}_{0}(z)y^{a}(z)-\hat{N}_{0}(z)u(z)=a_{2}(z)+r_{0,n}(z),
\label{eq4-26} \\
a_{2}(z)=\hat{M}_{0}(z)a_{y}(z)+\hat{N}_{0}(z)a_{u}(z),  \notag
\end{gather}%
where $r_{0,n}(z)$ describes the influence of the noises on the residuals $%
r_{0,p}(z)$ and $r_{0}(z)$ and is given by%
\begin{equation}
e(k+1)=\left( A-L_{0}C\right) e(k)+\omega (k)-L_{0}\nu
(k),r_{0,n}(k)=Ce(k)+\nu (k).
\end{equation}%
Hence, 
\begin{gather}
r(z)=P(z)\left[ 
\begin{array}{c}
r_{u,0}(z) \\ 
r_{0}(z)%
\end{array}%
\right] =\left[ 
\begin{array}{c}
r_{u}(z) \\ 
r_{0,K}(z)%
\end{array}%
\right] =\left[ 
\begin{array}{c}
a_{1}(z)-\bar{Q}_{\sigma }(z)a_{2}(z) \\ 
Q_{K0}(z)\left( a_{2}(z)+r_{0,n}(z)\right)%
\end{array}%
\right] \Longrightarrow  \notag \\
J(k)=\lambda r_{u}^{T}(k)r_{u}(k)+r_{0,K}^{T}(k)\Sigma _{r}^{-1}r_{0,K}(k) 
\notag \\
=\lambda \bar{a}_{1}^{T}(k)\bar{a}_{1}(k)+\left( \bar{a}_{2}(k)+r_{K}(k)%
\right) ^{T}\Sigma _{r}^{-1}\left( \bar{a}_{2}(k)+r_{K}(k)\right) \sim 
\mathcal{\chi }^{2}\left( \delta ,m\right) ,  \label{eq4-27} \\
\bar{a}_{1}(z)=a_{1}(z)-\bar{Q}_{\sigma }(z)a_{2}(z),\bar{a}%
_{2}(z)=Q_{K0}(z)a_{2}(z),r_{K}(z)=Q_{K0}(z)r_{0,n}(z),  \notag
\end{gather}%
where $\mathcal{\chi }^{2}\left( \delta ,m\right) $ denotes a noncentral $%
\mathcal{\chi }^{2}$ distribution with%
\begin{equation*}
\delta =\lambda \bar{a}_{1}^{T}(k)\bar{a}_{1}(k)+\bar{a}_{2}^{T}(k)\Sigma
_{r}^{-1}\bar{a}_{2}(k)
\end{equation*}%
as the noncentrality parameter and $m$ the degree of freedom. As well-known 
\cite{Ding2014}, the test statistic (\ref{eq4-4}) and the threshold (\ref%
{eq4-5}) lead to the maximal fault detectability and guarantee the FAR
bounded by $\alpha .$ Moreover, from (\ref{eq4-25}), (\ref{eq4-26}) and (\ref%
{eq4-27}) it can be evidently seen that all attacks, $a_{r_{en}},$ $a_{u},$ $%
a_{y},$ can be well detected as far as the dynamics of the encoded signal $%
r_{en,\sigma }$ or equivalently $\left( X_{\sigma },Y_{\sigma }\right) $ is
not identified.

\bigskip

As summary of the proposed detection scheme, the configuration of the
detection system including data transmissions is sketched in Figure 4.

\begin{figure}[h]
\centering\includegraphics[width=11cm,height=8cm]{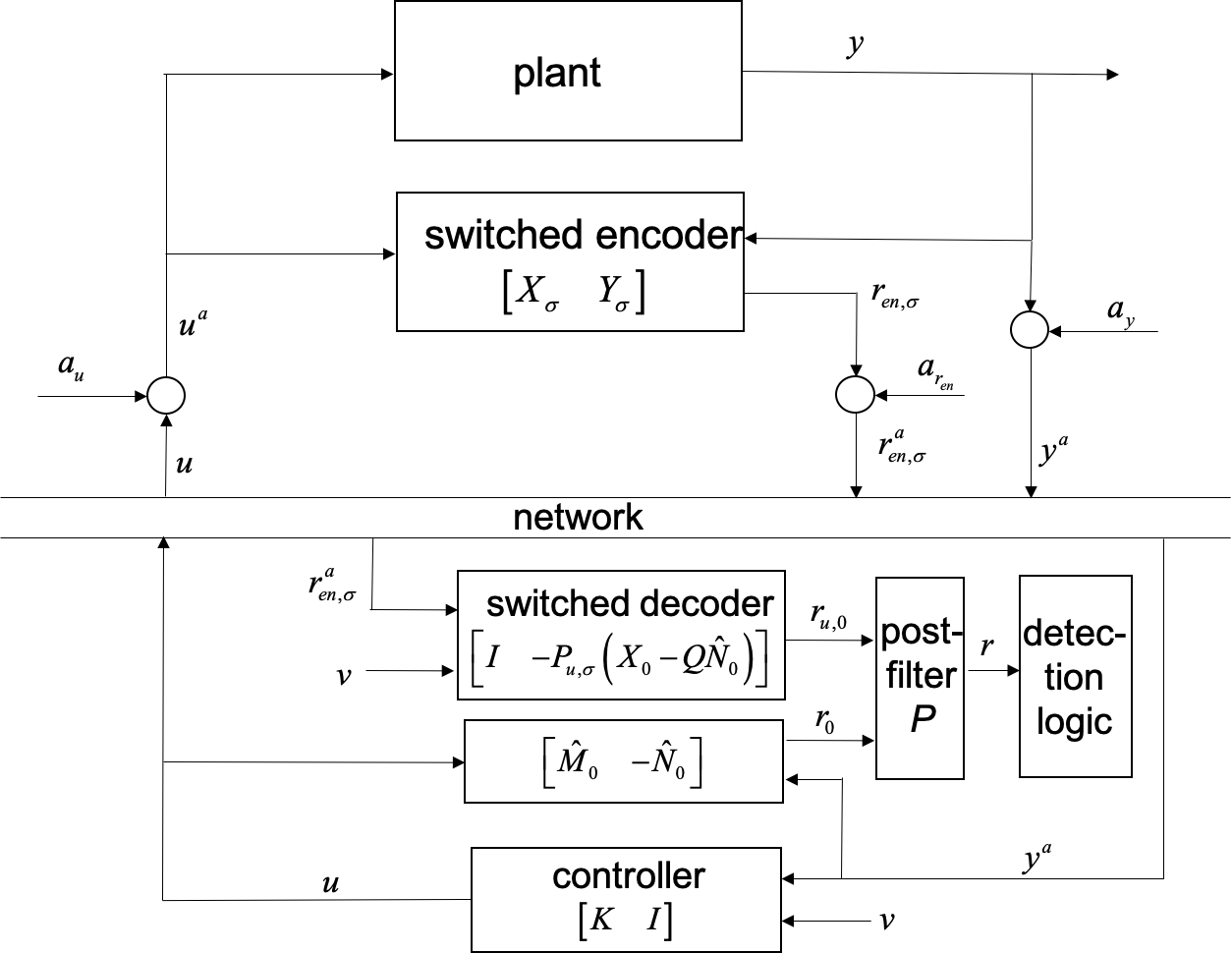}
\caption{Schematic description of the proposed attack detection system}
\end{figure}

At the end of this section, we would like to underline the following points:

\begin{itemize}
\item the detection scheme proposed in this section and based on the
residual signals $r_{0},r_{0,u}$ can be analogously realised as well using
the alternative residual signals $r_{u,c},r_{y,c}$ defined in (\ref{eq3-23});

\item the test statistic (\ref{eq4-4}) and the threshold (\ref{eq4-5})
deliver the optimal attack detection only on the assumptions of (i) the
statistic features of the noises being specified by (\ref{eq3-2a})-(\ref%
{eq3-2b}) , and (ii) the additive character of the kernel attacks being
under consideration \cite{Ding2020}, and

\item in case that the noises cannot be described by (\ref{eq3-2a})-(\ref%
{eq3-2b}) or/and the cyber-attacks are presented e.g. in multiplicative form
like false data injection attacks \cite{LZLWD2017,GWSOM2019}, sophisticated
detection schemes are needed. Some of these methods are reported in \cite%
{LD-Automatica-2020,Ding2020}.
\end{itemize}

\section{An encrypted configuration of feedback control and detection systems%
}

The basis for the execution of kernel attacks is that attackers have
knowledge of plant dynamics. Among numerous possibilities to gain such
information, eavesdropping attacks enable collecting sufficient amount of
process data which can then be used for identifying the plant dynamics. It
is state of the art that in real industrial applications plant input and
output data, $u(k)$ and $y(k),$ are often transmitted between the control
and monitoring station and the plant via networks. Such system
configurations make an identification of the plant dynamics considerably
easy. In this section, we propose an encrypted configuration scheme of
feedback control systems. The core of the alternatively configured control
systems consists in the transmission of encoded system signals, instead of $%
u(k)$ and $y(k),$ from which a direct identification of the plant dynamics
without \textit{a priori} knowledge becomes almost impossible. The basis for
this encrypted configuration is the so-called functionalisation of dynamic
controllers introduced in the unified framework of control and detection.

\subsection{Functionalisation of all stabilising feedback controllers}

Recall the observer-based realisation of all stabilising controllers given
in (\ref{eq2-13a})-(\ref{eq2-13c}). It can be divided into several
functional modules:

\begin{itemize}
\item an observer and an observer-based residual generator, 
\begin{gather*}
\hat{x}(k+1)=A\hat{x}(k)+Bu(k)+L_{0}r_{0}(k), \\
r_{0}(k)=y(k)-\hat{y}(k),\hat{y}(k)=C\hat{x}(k)+Du(k),
\end{gather*}%
which serve as an information provider for the controller and diagnostic
system, and deliver a state estimation, $\hat{x},$ and the primary residual, 
$r_{0}=y-\hat{y},$

\item control law%
\begin{equation*}
u(z)=F_{0}\hat{x}(z)-Q(z)r_{0}(z)+\hat{V}(z)v(z),
\end{equation*}%
including

\begin{itemize}
\item a feedback controller: $F_{0}\hat{x}(z)-Q(z)r_{0}(z)$ and

\item a feed-forward controller: $\hat{V}(z)v(z),\hat{V}=X_{0}-Q\hat{N}_{0},$
and in addition, for the detection purpose,
\end{itemize}

\item detector $R(z)r_{0}(z)$ with $R(z)$ as a stable post-filter.
\end{itemize}

This modular structure provides us with a clear parameterisation of the
functional modules:

\begin{itemize}
\item the state observer is parameterised by $L_{0},$

\item the feedback controller by $F_{0},Q,$

\item the feed-forward controller by $\hat{V},$ and

\item the detector by $R.$
\end{itemize}

Although all five parameters listed above are available for the design and
online optimisation objectives, they have evidently different
functionalities, as summarised below:

\begin{itemize}
\item $F_{0},L_{0}$ determine the stability and eigen-dynamics of the
closed-loop,

\item $R,\hat{V}$ have no influence on the system stability, and $R$ serves
for the optimisation of the detectability, while $\hat{V}$ for the tracking
behavior, and

\item $Q$ is used to enhance the system robustness and control performance.
The design and update of $Q$ will have influence on the system dynamics and
stability, when parameter uncertainties or degradations are present in the
system.
\end{itemize}

It is evident that the above five parameters have to be, due to their
different functionalities, treated with different priorities. Recall that
system stability and eigen-dynamics are the fundamental requirement on an
automatic control system. This requires that the system stability should be
guaranteed, also in case of cyber-attacks. Differently, $Q,R,\hat{V}$ are
used to optimise control or detection performance. In case that a temporary
system performance degradation is tolerable, the real-time demand and the
priority for an online optimisation of $Q,R,\hat{V}$ are relatively lower.
Under these considerations, we propose in the next sub-section an encrypted
control system configuration based on the above controller functionalisation.

\subsection{An encrypted system configuration scheme}

To begin with, we would like to emphasise that the objective of the system
configuration proposed in the sequel is to prevent system knowledge from
attackers in the manner that the plant model cannot be identified using the
data possibly collected by attackers by means of eavesdropping attacks.
Moreover, the basic requirements on the system control performance like the
stability are to be met.

\bigskip

The proposed encrypted system configuration mainly consists of

\begin{itemize}
\item on the plant side, an observer-based state feedback controller and
residual generator,%
\begin{align}
\hat{x}(k+1)& =A\hat{x}(k)+Bu(k)+L_{0}r_{0,p}(k),r_{0,p}(k)=y(k)-\hat{y}(k),
\notag \\
u(k)& =F_{0}\hat{x}(k)+\gamma (k)\Longrightarrow  \notag \\
\hat{x}(k+1)& =\left( A+BF_{0}\right) \hat{x}(k)+B\gamma (k)+L_{0}r_{0,p}(k)
\label{eq5-3a} \\
& =\left( A-L_{0}C\right) \hat{x}(k)+\left( B-L_{0}D\right) u(k)+L_{0}y(k),
\label{eq5-3b}
\end{align}%
where $\gamma $ is the signal (vector) received from the monitoring and
control side,

\item on the monitoring and control side,%
\begin{equation*}
\gamma (z)=\hat{V}(z)v(z)-Q(z)r_{0,p}(z),
\end{equation*}%
where $r_{0,p}$ is received from the plant side and $v$ is the reference
vector,

\item transmission from the plant side to the monitoring and control side, $%
r_{0,p}(k),$

\item transmission from the monitoring and control side to the plant side, $%
\gamma (k).$
\end{itemize}

Depending on applications, the following functional modules can be further
realised and integrated on the monitoring and control side, for instance,

\begin{itemize}
\item reconstructing $y(k),$%
\begin{gather}
\hat{x}(k+1)=\left( A+BF_{0}\right) \hat{x}(k)+B\gamma (k)+L_{0}r_{0,p}(k),
\label{eq5-1a} \\
y(k)=r_{0,p}(k)+\hat{y}(k)=\left( C+DF_{0}\right) \hat{x}(k)+D\gamma
(k)+r_{0,p}(k)  \label{eq5-1b}
\end{gather}%
with $r_{0,p}(k)$ received from the plant side,

\item tuning $Q$ using $r_{0,p}(k)$ and $v(k)$ to enhance the stability
margin, as reported in \cite{LLDYP-2019}, or

\item recovering control performance degradation using $y(k)$ and $u(k),$ as
described in \cite{Ding2020}.
\end{itemize}

It is evident that, according to the observer-based realisation of all
stabilising controllers, the control input $u(k)$ acted on the actuators
(located on the plant side) is given by%
\begin{equation*}
u(k)=F_{0}\hat{x}(k)+\gamma (k)\Longleftrightarrow u(z)=K(z)y(z)+v(z)
\end{equation*}%
with $K$ satisfying (\ref{eq2-12a})-(\ref{eq2-12b}). Different from the
standard system configuration, for instance the one shown in Figure 1, the
observer-based state feedback controller and residual generator (\ref{eq5-3a}%
)-(\ref{eq5-3b}) running on the plant side serve as

\begin{itemize}
\item an encoder for an encrypted transmission of the plant measurement $%
y(k),$ i.e. $r_{0,p}(k)$ instead of $y(k),$

\item a decoder for control input $u(k)=F_{0}\hat{x}(k)+\gamma (k),$ and

\item a local controller guaranteeing the basic control performance like the
stability even if the communication between the both sides of the control
system is considerably attacked.
\end{itemize}

Simultaneously, the recovering algorithm (\ref{eq5-1a})-(\ref{eq5-1b})
running on the monitoring and control side acts (i) as a decorder for $y$
and (ii) $\gamma =\hat{V}v-Qr_{0,p}$ as an encoder for an encrypted
transmission of the control signal from the monitoring and control side to
the plant.

\bigskip

Considering that, during the attack-free operation, $r_{0,p}$ is noise (and
even white noise when $L_{0}$ is set to be the Kalman filter gain matrix),
it is obviously impossible to identify the plant model $G_{u}$ by means of $%
r_{0,p}$ and $\gamma $ that could be eavesdropped during their transmission.
As a result, it can be claimed that the encrypted control system
configuration proposed in this sub-section fully fulfills the design
requirements.

\subsection{The associated attack detection scheme}

Figure 5 sketches schematically the proposed encrypted system configuration. On the assumptions that

\begin{itemize}
\item the control loop under consideration is configurated as sketched in
Figure 5,

\item the attacker has no knowledge about the plant model $G_{u},$ and

\item both $\gamma $ and $r_{0,p}$ are corrupted by the attack signals $%
a_{\gamma }$ and $a_{r_{0}}$ respectively, i.e. 
\begin{align*}
\gamma ^{a}(k)& =\gamma (k)+a_{\gamma }(k)\Longrightarrow u^{a}(k)=\gamma
(k)+a_{\gamma }(k)+F_{0}\hat{x}(k), \\
r_{0}^{a}(k)& =r_{0,p}(k)+a_{r_{0}}(k),
\end{align*}
\end{itemize}

we propose the following attack detection scheme performed on the monitoring
and control side. Similar to the controller, the attack detector is also
distributedly realised on the both sides of the control system. 

\begin{figure}[h]
\centering\includegraphics[width=14cm,height=10cm]{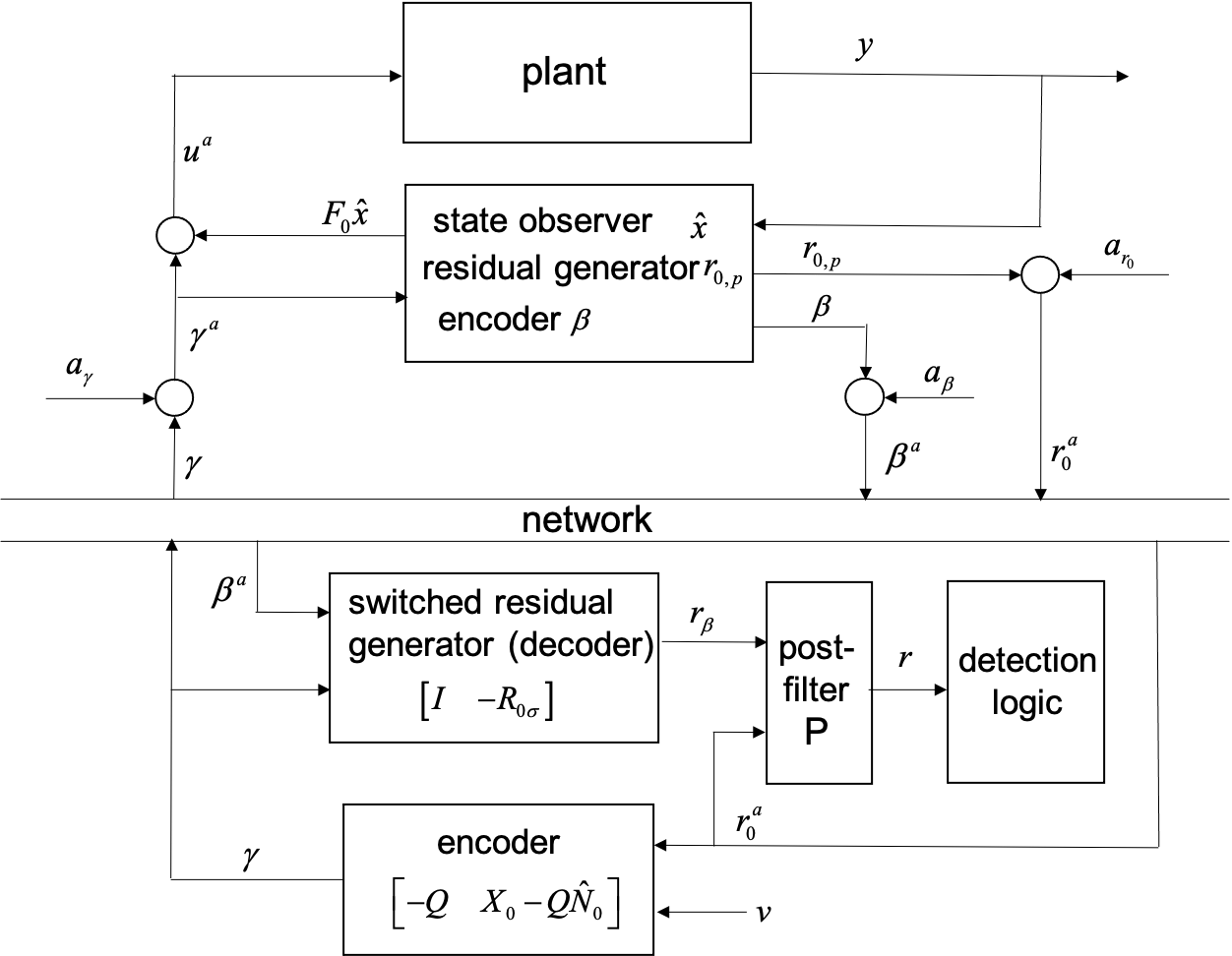}
\caption{Schematic description of the encrypted control and detection system
configuration}
\end{figure}

\bigskip

Remember that in the attack-free case%
\begin{gather}
X_{0}(z)u(z)+Y_{0}(z)y(z)-\gamma (z)  \notag \\
=X_{0}(z)u(z)+Y_{0}(z)y(z)-\left( \bar{v}(z)-Q(z)r_{0,p}(z)\right)  \notag \\
=u(z)-\left( F_{0}\hat{x}(z)+\gamma (z)\right) =0.  \label{eq5-5}
\end{gather}%
It motivates us to encrypt the detector as follows. At first, the encoded
signal $\beta (k)$ is generated on the plant side,%
\begin{equation}
\beta (k)=F_{0}\hat{x}(k)-F_{\sigma }\hat{x}(k),  \label{eq5-7e}
\end{equation}%
where $F_{\sigma }$ is a switched feedback gain introduced in the previous
section. It follows from Lemma \ref{Le4-1} and Theorem \ref{Theo4-3} that 
\begin{gather}
\beta (z)=\left( F_{0}-F_{\sigma }\right) \hat{x}(z)=u^{a}(z)-F_{\sigma }%
\hat{x}(z)-\left( u^{a}(z)-F_{0}\hat{x}(z)\right)  \notag \\
=R_{0\sigma }(z)\left( X_{0}(z)u^{a}(z)+Y_{0}(z)y(z)\right) +Q_{0\sigma
}(z)r_{0,p}(z),  \label{eq5-7} \\
R_{0\sigma }(z)=P_{u,\sigma }(z)-I=\left( F_{0}-F_{\sigma }\right) \left(
zI-A_{F_{0}}\right) ^{-1}B,  \label{eq5-7d} \\
Q_{0\sigma }(z)=\left( F_{0}-F_{\sigma }\right) \left( zI-A_{F_{0}}\right)
^{-1}L_{0}.  \notag
\end{gather}%
Here, $P_{u,\sigma }(z)$ is given in Theorem \ref{Theo4-3}. The encoded
signal $\beta $ is then sent to the monitoring and control side, at which a
residual signal is generated by decoding $\beta $ as follows%
\begin{equation}
r_{\beta }(z)=\beta ^{a}(z)-R_{0\sigma }(z)\gamma (z) ,
\label{eq5-8}
\end{equation}%
where 
\begin{equation*}
\beta ^{a}(k)=\beta (k)+a_{\beta }(k)
\end{equation*}%
denotes the corrupted signal $\beta $ due to the cyber-attack $a_{\beta }.$
It turns out, remembering (\ref{eq5-5}), 
\begin{align}
r_{\beta }(z)& =a_{\beta }(z)+R_{0\sigma }(z)\left(
X_{0}(z)u^{a}(z)+Y_{0}(z)y(z)-\gamma (z)\right) +Q_{0\sigma
}(z)r_{0,p}(z)  \notag \\
& =a_{\beta }(z)+R_{0\sigma }(z)X_{0}(z)a_{\gamma }(z)+Q_{0\sigma
}(z)r_{0,n}(z).
\end{align}%
with $r_{0,n}$ denoting the influence of the noises on the residual vector.
Therefore, it holds, on the monitoring and control side, 
\begin{equation}
\left[ 
\begin{array}{c}
r_{\beta }(z) \\ 
r_{0}^{a}(z)%
\end{array}%
\right] =\left[ 
\begin{array}{ccc}
I & R_{0\sigma }(z)X_{0}(z) & 0 \\ 
0 & 0 & I%
\end{array}%
\right] \left[ 
\begin{array}{c}
a_{\beta }(z) \\ 
a_{\gamma }(z) \\ 
a_{r_{0}}(z)%
\end{array}%
\right] +\left[ 
\begin{array}{c}
Q_{0\sigma }(z) \\ 
I%
\end{array}%
\right] r_{0,n}(z).  \label{eq5-8a}
\end{equation}%
As a result, we have

\begin{Theo}
\label{Theo4-4}Given the plant model (\ref{eq3-1a})-(\ref{eq3-1b}), the
control law $K$ satisfying (\ref{eq2-12a})-(\ref{eq2-12b}) and residuals $%
r_{\beta }$ and $r_{0}^{a}$ that are realised in the encrypted system
configuration shown in Figure 5, the attacks $a_{\beta },a_{\gamma }$ and $%
a_{r_{0}}$ are stealthy, if and only if the conditions,%
\begin{equation}
a_{r_{0}}(k)=0,a_{\beta }(z)+R_{0\sigma }(z)X_{0}(z)a_{\gamma }(z)=0,
\label{eq5-10}
\end{equation}%
are satisfied.
\end{Theo}

Theorem \ref{Theo4-4} reveals that

\begin{itemize}
\item an (additive) attack on the residual signal $r_{0,p}$ can be
(structurally) directly detected, and

\item keeping $a_{\beta },a_{\gamma }$ stealthy is almost impossible, since
condition (\ref{eq5-8}) can hardly be satisfied, (i) without system
knowledge, (ii) without knowing the purpose of using and transmissing $\beta 
$ and $\gamma ,$ and (iii) in particular when $R_{0\sigma }(z)$ is a
switched system.
\end{itemize}

\subsection{Implementation of the control and detection systems}

Now, we summarise the implementation issues of the proposed control and
detection systems.

\bigskip

On the plant side, the state observer (\ref{eq5-3a}) (equivalently (\ref%
{eq5-3b})) builds the core of the system implementation. Based on the state
estimate $\hat{x}(k),$ the control input $u^{a}(k),$ the residual signal $%
r_{0,p}(k)$ as well as the encoded signal $\beta (k)$ are formed, 
\begin{align}
u^{a}(k)& =F_{0}\hat{x}(k)+\gamma ^{a}(k),  \label{eq5-7a} \\
r_{0,p}(k)& =y(k)-C\hat{x}(k)-Du^{a}(k),  \label{eq5-7b} \\
\beta (k)& =\left( F_{0}-F_{\sigma }\right) \hat{x}(k).  \label{eq5-7c}
\end{align}%
For running the realisation algorithms, the system on the plant side
receives the signal $\gamma ^{a}$ from the monitoring and control side. It
sends the residual signal $r_{0,p}$ and encoded signal $\beta $ to the
system running on the monitoring and control side. It is of considerable
interest to remark that the state observer (\ref{eq5-3a}) serves both as a
decoder for the control signal, as given in (\ref{eq5-7a}), and as an
encoder for the controller and for the generation of residual signal $%
r_{\beta }$ (that are implemented on the monitoring and control side), as
described by (\ref{eq5-7b}) and (\ref{eq5-7c}).

\bigskip

On the monitoring and control side, $\gamma (k)$ is first computed as
follows 
\begin{gather}
x_{v}(k+1)=\left( A-L_{0}C\right) x_{v}(k)+(B-L_{0}D)v(k),  \label{eq5-9a} \\
\gamma (z)=v(z)-F_{0}x_{v}(z)-Q(z)\left( Cx_{v}(z)+Dv(z)-r_{0}^{a}(z)\right)
.  \label{eq5-9b}
\end{gather}%
Then, $r_{\beta }$ is generated as%
\begin{align}
x_{\beta }(k+1)& =\left( A+BF_{0}\right) x_{\beta }(k)+B\gamma
(k) ,  \label{eq5-10a} \\
r_{\beta }(k)& =\beta ^{a}(k)-\left( F_{\sigma }-F_{0}\right) x_{\beta }(k).
\label{eq5-10b}
\end{align}%
It is worth emphasising that computation (\ref{eq5-9a})-(\ref{eq5-9b})
serves as an encoder for the control signal, while the system (\ref{eq5-10a}%
)-(\ref{eq5-10b}) acts as a decoder.

\bigskip

Next, for detecting attacks $a_{\beta },a_{\gamma }$ and $a_{r_{0}}$
optimally, the residual vector 
\begin{equation}
r(z)=\left[ 
\begin{array}{cc}
I & -Q_{0\sigma }(z) \\ 
0 & Q_{K0}(z)%
\end{array}%
\right] \left[ 
\begin{array}{c}
r_{\beta }(z) \\ 
r_{0}^{a}(z)%
\end{array}%
\right] =:\left[ 
\begin{array}{c}
r_{u}(z) \\ 
r_{0,K}(z)%
\end{array}%
\right]  \label{eq5-6}
\end{equation}%
and the test statistic 
\begin{equation*}
J(k)=\lambda r_{u}^{T}(k)r_{u}(k)+r_{0,K}^{T}(k)\Sigma _{r}^{-1}r_{0,K}(k)
\end{equation*}%
are built with $Q_{K0}$ as given in (\ref{eq4-24a}), which is analogue to
the result described in Sub-section \ref{sub-sec4-4}. We have 
\begin{equation}
J(k)=\lambda r_{u}^{T}(k)r_{u}(k)+r_{0,K}^{T}(k)\Sigma
_{r}^{-1}r_{0,K}(k)\sim \mathcal{\chi }^{2}\left( m\right) ,  \label{eq4-4a}
\end{equation}%
and thus the threshold is set to be%
\begin{equation}
J_{th}=\mathcal{\chi }_{\alpha }^{2}\left( m\right)  \label{eq4-5a}
\end{equation}%
for a given upper-bound of false alarm rate $\alpha $. In case of attacks, 
\begin{gather}
J(k)=\lambda r_{u}^{T}(k)r_{u}(k)+r_{0,K}^{T}(k)\Sigma _{r}^{-1}r_{0,K}(k) 
\notag \\
=\lambda a_{1}^{T}(k)a_{1}(k)+\left( a_{2}(k)+r_{K}(k)\right) ^{T}\Sigma
_{r}^{-1}\left( a_{2}(k)+r_{K}(k)\right) \sim \mathcal{\chi }^{2}\left(
\delta ,m\right) \\
a_{1}(z)=a_{\beta }(z)+R_{0\sigma }(z)X_{0}(z)a_{\gamma }(z)-Q_{0\sigma
}(z)a_{r_{0}}(z),  \notag \\
a_{2}(z)=Q_{K0}(z)a_{r_{0}}(z),r_{K}(k)\sim \mathcal{N}\left( 0,\Sigma
_{r}\right) ,
\end{gather}%
where $\mathcal{\chi }^{2}\left( \delta ,m\right) $ denotes a noncentral $%
\mathcal{\chi }^{2}$ distribution with%
\begin{equation*}
\delta =\lambda a_{1}^{T}(k)a_{1}(k)+a_{2}^{T}(k)\Sigma _{r}^{-1}a_{2}(k)
\end{equation*}%
as the noncentrality parameter and $m$ the degree of freedom.

\section{Examples and experimental study}

\subsection{Examples of detecting typical kernel attacks}

As examples, we will demonstrate that the zero dynamics, covert and replay
attacks as kernel attacks\ can be well detected using the detection schemes
proposed in Sections 4 and 5.

\begin{Exp}
Consider a zero dynamics attack satisfying (\ref{eq2-11b}). For our purpose
of detecting $a_{u},$ applying both detection schemes presented in
Sub-sections \ref{Sec4-3}-\ref{sub-sec4-4} and Section 5 results in

\begin{itemize}
\item by detector (\ref{eq4-24}) whose dynamics with respect to the
(possible) attack signals is described by (\ref{eq4-27}):%
\begin{equation}
r(z)=\left[ 
\begin{array}{c}
r_{u}(z) \\ 
r_{0,K}(z)%
\end{array}%
\right] =\left[ 
\begin{array}{c}
P_{u,\sigma }(z)X_{0}(z)a_{u}(z)+a_{r_{en}}(z) \\ 
r_{K}(z)%
\end{array}%
\right] ,  \label{eq6-1}
\end{equation}%
where $a_{r_{en}}$ denotes the (possible) attack signal on the transmitted
signal $r_{en}$ that builds an (encoded) part of $r_{u},$

\item by detector (\ref{eq5-6}) whose dynamics with respect to the
(possible) attack signals is described by (\ref{eq5-8a}): 
\begin{equation}
r(z)=\left[ 
\begin{array}{c}
r_{u}(z) \\ 
r_{0,K}(z)%
\end{array}%
\right] =\left[ 
\begin{array}{c}
a_{\beta }(z)+R_{0\sigma }(z)X_{0}(z)a_{\gamma }(z) \\ 
r_{K}(z)%
\end{array}%
\right] ,a_{u}(z)=a_{\gamma }(z)  \label{eq6-2}
\end{equation}%
with the (additional) attack signal $a_{\beta }$ on the encoded signal $%
\beta .$
\end{itemize}

It is evident that in the former case, $a_{u}$ can be detected using $r$ as
far as the attacker could not identify $P_{u,\sigma }$ or equivalently $%
X_{\sigma }$ and thus set $a_{r_{en}}$ equal to $-P_{u,\sigma }X_{0}a_{u}.$
For the latter case, as long as the switched system $R_{0\sigma }X_{0}$
could not be identified, it is impossible for the attacker to construct $%
a_{\beta }$ equal to $-R_{0\sigma }X_{0}a_{u}.$ Consequently, both $a_{\beta
}$ and $a_{u}$ can be detected. We would like to emphasise that in this case
it is impossible to identify the plant dynamics $\left( \hat{M}_{0},\hat{N}%
_{0}\right) $ using eavesdropped data $\gamma (k),r_{0,p}(k).$
\end{Exp}

\begin{Exp}
Now, consider the both detection systems under a covert attack satisfying (%
\ref{eq3-13}). It holds,

\begin{itemize}
\item by detector (\ref{eq4-24}):%
\begin{equation}
r(z)=\left[ 
\begin{array}{c}
r_{u}(z) \\ 
r_{0,K}(z)%
\end{array}%
\right] =\left[ 
\begin{array}{c}
P_{u,\sigma }(z)\left( X_{0}(z)a_{u}(z)-Y_{0}(z)a_{y}(z)\right)
+a_{r_{en}}(z) \\ 
r_{K}(z)%
\end{array}%
\right] ,  \label{eq6-3}
\end{equation}%
with the (possible) additional attack $a_{r_{en}}$ on the transmitted signal 
$r_{en},$

\item by detector (\ref{eq5-6}): 
\begin{align}
r(z)& =\left[ 
\begin{array}{c}
r_{u}(z) \\ 
r_{0,K}(z)%
\end{array}%
\right] =\left[ 
\begin{array}{c}
a_{\beta }(z)+R_{0\sigma }(z)X_{0}(z)a_{\gamma }(z)-Q_{0\sigma
}(z)a_{r_{0}}(z) \\ 
Q_{K0}(z)a_{r_{0}}(z)%
\end{array}%
\right] ,  \label{eq6-4} \\
a_{u}(z)& =a_{\gamma }(z),a_{y}(z)=a_{r_{0}}(z)  \notag
\end{align}%
where it is assumed that the attack signal $a_{y}$ is added to the
transmitted signal $r_{0,p},$ since $r_{0,p}$ instead of $y$ is transmitted
from the plant side to the monitoring and control side.
\end{itemize}

It is clear that in the first case, $a_{u}$ and $a_{y}$ can be detected as
far as the attacker could not identify $\left( X_{\sigma },Y_{\sigma
}\right) $. It is of considerable interest to notice the results in the
second case. Using the detector (\ref{eq5-6}), we can identify the attack $%
a_{y}$ ($a_{r_{0}}$)$,$%
\begin{equation*}
a_{r_{0}}(z)=a_{y}(z)=Q_{K0}^{-1}(z)r_{0,K}(z),Q_{K0}^{-1}(z)=I+C\left(
zI-A+L_{0}C\right) ^{-1}\left( L_{K}-L_{0}\right) ,
\end{equation*}%
and moreover estimate $a_{u}$ based on%
\begin{equation*}
\hat{M}_{0}(z)a_{y}(z)+\hat{N}_{0}(z)a_{u}(z)=0\Longleftrightarrow \hat{N}%
_{0}(z)a_{u}(z)=-\hat{M}_{0}(z)Q_{K0}^{-1}(z)r_{0,K}(z),
\end{equation*}%
when $\left( a_{y},a_{u}\right) $ is a covert attack. In this case, $%
a_{\beta }$ can also be estimated in terms of 
\begin{equation*}
a_{\beta }(z)=-R_{0\sigma }(z)a_{u}(z)+Q_{0\sigma }(z)a_{y}(z).
\end{equation*}%
Finally, as far as $R_{0\sigma }(z),Q_{0\sigma }(z)$ are not identified, any
attacks of $a_{\beta },a_{y},a_{u}$ can be detected. This example clearly
demonstrates the advantage of the detector (\ref{eq5-6}) over the detector (%
\ref{eq4-24}) and other reported attack detectors.
\end{Exp}

\begin{Exp}
We now address the detection issue of replay attacks under the assumption of
steady operation, i.e. 
\begin{equation}
y(k)\approx y(k-i),u(k)=u(k-i),i=1,\cdots .  \label{eq6-5}
\end{equation}%
Since in our detection schemes proposed in the last two sections additional
signals, $r_{en}$ and $\beta ,$ are transmitted from the plant side to the
monitoring and control side, it is assumed that the attacker has collected
and saved the (attack-free) data $r_{en}(j),\beta (j),j\in \left[
k_{0},k_{0}+M\right] .$ When the data are replayed over the time interval $%
\left[ k,k+M\right] ,k>k_{0}+M,$ it holds 
\begin{equation}
r_{en}^{a}(i)=r_{en}(i-(k-k_{0})),\beta ^{a}(i)=\beta (i-(k-k_{0})),i\in 
\left[ k,k+M\right] .  \label{eq6-6}
\end{equation}%
Moreover, an attack signal on the actuators is injected, for instance,%
\begin{equation*}
a_{u}(i)=a_{\gamma }(i),i\in \left[ k,k+M\right] .
\end{equation*}%
It turns out

\begin{itemize}
\item by detector (\ref{eq4-24}):%
\begin{gather}
r(z)=\left[ 
\begin{array}{c}
r_{u}(z) \\ 
r_{0,K}(z)%
\end{array}%
\right] \approx \left[ 
\begin{array}{c}
\Delta r_{en}(z)-\Delta r_{K}(z) \\ 
r_{K}(z)%
\end{array}%
\right] ,  \label{eq6-7} \\
\Delta r_{en}(i)=r_{en,\sigma \left( i-(k-k_{0})\right)
}(i-(k-k_{0}))-r_{en,\sigma \left( i\right) }(i),  \notag \\
r_{en,\sigma \left( i-(k-k_{0})\right) }(z)=X_{\sigma \left(
i-(k-k_{0})\right) }(z)u(z^{-(k-k_{0})})+Y_{\sigma \left( i-(k-k_{0})\right)
}(z)y(z^{-(k-k_{0})}),  \notag \\
r_{en,\sigma \left( i\right) }(z)=X_{\sigma \left( i\right)
}(z)u(z)+Y_{\sigma \left( i\right) }(z)y(z),  \notag \\
\Delta r_{K}(i)=r_{K,\sigma \left( i-(k-k_{0})\right)
}(i-(k-k_{0}))-r_{K,\sigma \left( i\right) }(i),  \label{eq6-7a} \\
r_{K,\sigma \left( i-(k-k_{0})\right) }(z)=\bar{Q}_{\sigma \left(
i-(k-k_{0})\right) }(z)r_{K}(z^{-(k-k_{0})}),r_{K,\sigma \left( i\right)
}(z)=\bar{Q}_{\sigma \left( i\right) }(z)r_{K}(z),  \label{eq6-7b}
\end{gather}%
due to assumptions (\ref{eq6-5}) and (\ref{eq6-6}),

\item by detector (\ref{eq5-6}): 
\begin{gather}
r(z)=\left[ 
\begin{array}{c}
r_{u}(z) \\ 
r_{0,K}(z)%
\end{array}%
\right] \approx \left[ 
\begin{array}{c}
\Delta _{\beta }(z)-\Delta r_{0,K}(z) \\ 
r_{K}(z)%
\end{array}%
\right] ,  \label{eq6-8} \\
\Delta _{\beta }(i)=\beta (i-(k-k_{0}))-\beta (i),  \notag \\
\Delta _{\beta }(z)=R_{0\sigma (i-(k-k_{0}))}(z)\left( X_{0}(z)u(z^{-\left(
k-k_{0}\right) })+Y_{0}(z)y(z^{-\left( k-k_{0}\right) })\right)  \notag \\
-R_{0\sigma (i)}(z)\left( X_{0}(z)u(z)+Y_{0}(z)y(z)\right) ,  \label{eq6-8a}
\\
\Delta r_{0,K}(z)=Q_{0\sigma (i-(k-k_{0}))}(z)r_{K}(z^{-\left(
k-k_{0}\right) })-Q_{0\sigma (i)}(z)r_{K}(z).  \label{eq6-8b}
\end{gather}
\end{itemize}

Now, we study dynamics (\ref{eq6-7}) and (\ref{eq6-8}). In the first case,
since $\left( X_{\sigma \left( i-(k-k_{0})\right) },Y_{\sigma \left(
i-(k-k_{0})\right) }\right) $ and $\left( X_{\sigma \left( i\right)
},Y_{\sigma \left( i\right) }\right) $ as well as $\bar{Q}_{\sigma \left(
i-(k-k_{0})\right) }$ and $\bar{Q}_{\sigma \left( i\right) }$ are generally
different, which leads to 
\begin{eqnarray*}
\Delta r_{en} &\approx &\left( X_{\sigma \left( i-(k-k_{0})\right)
}-X_{\sigma \left( i\right) }\right) u+\left( Y_{\sigma \left(
i-(k-k_{0})\right) }-Y_{\sigma \left( i\right) }\right) y\neq 0, \\
\Delta r_{K} &\approx &\left( \bar{Q}_{\sigma \left( i-(k-k_{0})\right) }-%
\bar{Q}_{\sigma \left( i\right) }\right) r_{K}\neq 0.
\end{eqnarray*}%
Consequently, both the mean and co-variance matrix of $r_{u}(z)$ will
change, which can be well detected using the generalised likelihood ratio
(GLR) method \cite{Ding2020}. The second case is similar to the first one so
that the replay attack can be detected in general, thanks to the fact that 
\begin{align*}
R_{0\sigma (i-(k-k_{0}))}\left[ 
\begin{array}{cc}
X_{\sigma \left( i-(k-k_{0})\right) } & \text{ }Y_{\sigma \left(
i-(k-k_{0})\right) }%
\end{array}%
\right] & \neq R_{0\sigma (i)}\left[ 
\begin{array}{cc}
X_{\sigma \left( i\right) } & \text{ }Y_{\sigma \left( i\right) }%
\end{array}%
\right] , \\
Q_{0\sigma (i-(k-k_{0}))}& \neq Q_{0\sigma (i)}.
\end{align*}
\end{Exp}

In comparison with the existing detection methods, it is clear that

\begin{itemize}
\item the two detection schemes proposed in this work guarantee structural
detection of any kernel attacks, while the most existing methods can be
generally applied to detecting a special type of kernel attacks;

\item in particular, both methods deliver reliable detection of replay
attacks without adding (additional) signals like a watermark in $u$ \cite%
{Mo2015-Watermarked-detection}. In fact, $\Delta r_{en}-\Delta r_{K}$ and $%
\Delta _{\beta }-\Delta r_{0,K}$ delivered by the detectors (\ref{eq4-24})
and (\ref{eq5-6}), respectively, act like a watermark but without any
influence on the control performance;

\item the design of both detectors are straightforward without complicated
computations, and

\item the required online computations are less demanding.
\end{itemize}

\subsection{Experimental study}
Experimental study on detecting cyber-attacks on a real three-tank control system is running and the achieved results will be reported.

\section{Conclusions}

In this work, we have studied issues of detecting stealthy integrity
cyber-attacks in the unified control and detection framework. The first
effort has been dedicated to the general form of integrity cyber-attacks
that cannot be detected using the well-established observer-based detection
technique. It has been demonstrated that any attacks lying in the system
kernel space cannot be detected by an observer-based detection system.
Correspondingly, the concept of kernel attacks has been introduced. The
replay, zero dynamics and covert attacks which are widely investigated in
the literature are the examples of kernel attacks. Our further effort has
been focused on the existence conditions of stealthy integrity attacks. To
this end, the unified framework of control and detection has applied. It has
been revealed that all kernel attacks can be structurally detected when
residual generation is extended to the space spanned by the control signal.
In other words, not only the observer-based residual, but also the control
signal based residual signals are needed for a reliable detection of kernel
attacks. As a result of this work, the necessary and sufficient conditions
for detecting kernel attacks are given.

\bigskip

Based on the analytical results in the first part of our study, we have
proposed two schemes for detecting kernel attacks. Using the known results
and methods of the unified control and detection framework, both schemes
result in reliable detection of kernel attacks without any loss of control
performance. While the first detector is configured similar to the existing
methods like the moving target method and auxiliary system aided detection
scheme \cite{MT-method-CDC2015,Zhang-CDC2017,DIBAJI2019-survey,GWSOM2019},
the second detector is realised with the encrypted transmissions of control
and monitoring signals in the feedback control system that prevent adversary
to gain system knowledge by means of eavesdropping attacks. The theoretical
basis for such detector configurations is the observer-based,
residual-driven realisation of all stabilising feedback controllers. In
particular, the functionalisation of controllers in the unified control and
detection framework plays an essential role in developing the second
detection scheme.

\bigskip

It should be remarked that our study in this work has been performed on the
assumptions that (i) the LTI system models are not corrupted with model
uncertainties, and (ii) the kernel attacks are presented in the additive
form (although the replay attack is a multiplicative signal, it is handled
as an additive one). In this context, the concept of kernel attacks and the
derived existence conditions are in fact the expressions of structural
properties of the feedback control system under consideration. So far, the
proposed detection schemes would work well in laboratory conditions, but
cannot be directly applied in real industrial applications without
modifications. This fact motivates our future work to deal with
cyber-attacks in the multiplicative form, for instance, false data injection
attacks \cite{LZLWD2017,GWSOM2019}, and on automatic control systems with
uncertainties. The unified control and detection framework and the
associated detection methods developed recently \cite%
{LD-Automatica-2020,Ding2020} could serve as efficient tools.

\bigskip

\textbf{Appendix Proof of Lemma 1}

\bigskip

Since 
\begin{align*}
R_{12}(z)F_{2}& =F_{2}+\left( F_{2}-F_{1}\right) \left( zI-A_{F_{2}}\right)
^{-1}BF_{2} \\
& =F_{1}+\left( F_{2}-F_{1}\right) +\left( F_{2}-F_{1}\right) \left(
zI-A_{F_{2}}\right) ^{-1}BF_{2} \\
& =F_{1}+\left( F_{2}-F_{1}\right) \left( zI-A_{F_{2}}\right) ^{-1}\left(
zI-A\right) ,
\end{align*}%
it turns out%
\begin{align*}
R_{12}(z)Y_{2}(z)& =-\left( F_{1}+\left( F_{2}-F_{1}\right) \left(
zI-A_{F_{2}}\right) ^{-1}\left( zI-A\right) \right) \left(
zI-A_{L_{2}}\right) ^{-1}L_{2}, \\
R_{12}(z)X_{2}(z)& =R_{12}(z)-\left( F_{1}+\left( F_{2}-F_{1}\right) \left(
zI-A_{F_{2}}\right) ^{-1}\left( zI-A\right) \right) \left(
zI-A_{L_{2}}\right) ^{-1}(B-L_{2}D).
\end{align*}%
Moreover, the relation%
\begin{gather*}
\left( zI-A\right) \left( zI-A_{L}\right) ^{-1}L=\left( I+LC\left(
zI-A\right) ^{-1}\right) ^{-1}L \\
=L\left( I-C\left( zI-A+LC\right) ^{-1}L\right) =L\hat{M}(z)
\end{gather*}%
leads to%
\begin{gather}
R_{12}(z)Y_{2}(z)=-F_{1}\left( zI-A_{L_{2}}\right) ^{-1}L_{2}+\bar{R}_{12}(z)%
\hat{M}_{2}(z)  \label{eq4-15a} \\
R_{12}(z)X_{2}(z)=R_{12}(z)-F_{1}\left( zI-A_{L_{2}}\right) ^{-1}(B-L_{2}D)-
\notag \\
\left( F_{2}-F_{1}\right) \left( zI-A_{F_{2}}\right) ^{-1}\left(
I-L_{2}C\left( zI-A+L_{2}C\right) ^{-1}\right) (B-L_{2}D)  \notag \\
=I-F_{1}\left( zI-A_{L_{2}}\right) ^{-1}(B-L_{2}D)+\left( F_{2}-F_{1}\right)
\left( zI-A_{F_{2}}\right) ^{-1}L_{2}D  \notag \\
+\left( F_{2}-F_{1}\right) \left( zI-A_{F_{2}}\right) ^{-1}L_{2}C\left(
zI-A+L_{2}C\right) ^{-1}(B-L_{2}D)  \notag \\
=I-F_{1}\left( zI-A_{L_{2}}\right) ^{-1}(B-L_{2}D)-\bar{R}_{12}(z)\hat{N}%
_{2}(z).  \label{eq4-15b}
\end{gather}%
Next, we consider $\left( zI-A_{L_{2}}\right) ^{-1}L_{2}$ and $\left(
zI-A_{L_{2}}\right) ^{-1}(B-L_{2}D).$ It is straightforward that 
\begin{gather}
\left( zI-A_{L_{2}}\right) ^{-1}L_{2}-\left( zI-A_{L_{1}}\right) ^{-1}L_{1} 
\notag \\
=\left( zI-A_{L_{2}}\right) ^{-1}\left( L_{2}-\left( zI-A+L_{2}C\right)
\left( zI-A_{L_{1}}\right) ^{-1}L_{1}\right)  \notag \\
=\left( zI-A_{L_{2}}\right) ^{-1}\left( L_{2}\left( I-C\left(
zI-A_{L_{1}}\right) ^{-1}L_{1}\right) -\left( zI-A\right) \left(
zI-A_{L_{1}}\right) ^{-1}L_{1}\right)  \notag \\
=\left( zI-A_{L_{2}}\right) ^{-1}L_{2}\hat{M}_{1}(z)-\left(
zI-A_{L_{2}}\right) ^{-1}L_{1}\hat{M}_{1}(z)\Longrightarrow  \notag \\
\left( zI-A_{L_{2}}\right) ^{-1}L_{2}=\left( zI-A_{L_{1}}\right)
^{-1}L_{1}+\left( zI-A_{L_{2}}\right) ^{-1}\left( L_{2}-L_{1}\right) \hat{M}%
_{1}(z)  \label{eq4-16}
\end{gather}%
as well as%
\begin{gather}
\left( zI-A_{L_{2}}\right) ^{-1}(B-L_{2}D)-\left( zI-A_{L_{1}}\right)
^{-1}(B-L_{1}D)=  \notag \\
\left( zI-A_{L_{2}}\right) ^{-1}\left( I-\left( zI-A+L_{2}C\right) \left(
zI-A_{L_{1}}\right) ^{-1}\right) B-\left( \left( zI-A_{L_{2}}\right)
^{-1}L_{2}-\left( zI-A_{L_{1}}\right) ^{-1}L_{1}\right) D  \notag \\
=\left( zI-A_{L_{2}}\right) ^{-1}\left( L_{1}-L_{2}\right) C\left(
zI-A_{L_{1}}\right) ^{-1}B-\left( zI-A_{L_{2}}\right) ^{-1}\left(
L_{2}-L_{1}\right) \hat{M}_{1}(z)D=  \notag \\
\left( zI-A_{L_{2}}\right) ^{-1}\left( L_{1}-L_{2}\right) \left( \hat{M}%
_{1}(z)D+C\left( zI-A_{L_{1}}\right) ^{-1}B\right) =-\left(
zI-A_{L_{2}}\right) ^{-1}\left( L_{2}-L_{1}\right) \hat{N}_{1}(z)  \notag \\
\Longrightarrow \left( zI-A_{L_{2}}\right) ^{-1}(B-L_{2}D)=\left(
zI-A_{L_{1}}\right) ^{-1}(B-L_{1}D)-\left( zI-A_{L_{2}}\right) ^{-1}\left(
L_{2}-L_{1}\right) \hat{N}_{1}(z).  \label{eq4-17}
\end{gather}%
Furthermore, it is well-known that 
\begin{align}
\left[ 
\begin{array}{cc}
-\hat{N}_{2}(z) & \hat{M}_{2}(z)%
\end{array}%
\right] & =Q_{21}(z)\left[ 
\begin{array}{cc}
-\hat{N}_{1}(z) & \hat{M}_{1}(z)%
\end{array}%
\right] \Longleftrightarrow  \label{eq4-18} \\
\left[ 
\begin{array}{cc}
-\hat{N}_{1}(z) & \hat{M}_{1}(z)%
\end{array}%
\right] & =Q_{12}(z)\left[ 
\begin{array}{cc}
-\hat{N}_{2}(z) & \hat{M}_{2}(z)%
\end{array}%
\right] .  \label{eq4-18a}
\end{align}%
Summarising (\ref{eq4-15a})-(\ref{eq4-18a}) leads to 
\begin{gather*}
R_{12}(z)Y_{2}(z)=Y_{1}(z)-\left( F_{1}\left( zI-A_{L_{2}}\right)
^{-1}\left( L_{2}-L_{1}\right) -\bar{R}_{12}(z)Q_{21}(z)\right) \hat{M}%
_{1}(z), \\
R_{12}(z)X_{2}(z)=X_{1}(z)+\left( F_{1}\left( zI-A_{L_{2}}\right)
^{-1}\left( L_{2}-L_{1}\right) -\bar{R}_{12}(z)Q_{21}(z)\right) \hat{N}%
_{1}(z)\Longrightarrow \\
\left[ 
\begin{array}{cc}
X_{1}(z) & Y_{1}(z)%
\end{array}%
\right] =R_{12}(z)\left[ 
\begin{array}{cc}
X_{2}(z) & Y_{2}(z)%
\end{array}%
\right] +\bar{Q}_{11}(z)\left[ 
\begin{array}{cc}
-\hat{N}_{1}(z) & \hat{M}_{1}(z)%
\end{array}%
\right]
\end{gather*}%
as well as%
\begin{gather*}
R_{12}(z)Y_{2}(z)=Y_{1}(z)-F_{1}\left( zI-A_{L_{2}}\right) ^{-1}\left(
L_{2}-L_{1}\right) \hat{M}_{1}(z)+\bar{R}_{12}(z)\hat{M}_{2}(z), \\
R_{12}(z)X_{2}(z)=X_{1}(z)+F_{1}\left( zI-A_{L_{2}}\right) ^{-1}\left(
L_{2}-L_{1}\right) \hat{N}_{1}(z)-\bar{R}_{12}(z)\hat{N}_{2}(z)%
\Longrightarrow \\
\left[ 
\begin{array}{cc}
X_{1}(z) & Y_{1}(z)%
\end{array}%
\right] =R_{12}(z)\left[ 
\begin{array}{cc}
X_{2}(z) & Y_{2}(z)%
\end{array}%
\right] +\bar{Q}_{12}(z)\left[ 
\begin{array}{cc}
-\hat{N}_{2}(z) & \hat{M}_{2}(z)%
\end{array}%
\right] , \\
\bar{Q}_{12}(z)=F_{1}\left( zI-A_{L_{2}}\right) ^{-1}\left(
L_{2}-L_{1}\right) Q_{12}(z)-\bar{R}_{12}(z).
\end{gather*}%
Since 
\begin{gather*}
\left( zI-A_{L_{2}}\right) ^{-1}\left( L_{2}-L_{1}\right) Q_{12}(z)=\left(
zI-A_{L_{2}}\right) ^{-1}\left( L_{2}-L_{1}\right) \left( I+C\left(
zI-A_{L_{1}}\right) ^{-1}\left( L_{2}-L_{1}\right) \right) \\
=\left( zI-A_{L_{2}}\right) ^{-1}\left( I+\left( L_{2}-L_{1}\right) C\left(
zI-A_{L_{1}}\right) ^{-1}\right) \left( L_{2}-L_{1}\right) =\left(
zI-A_{L_{1}}\right) ^{-1}\left( L_{2}-L_{1}\right) ,
\end{gather*}%
we finally have%
\begin{equation*}
\bar{Q}_{12}(z)=F_{1}\left( zI-A_{L_{1}}\right) ^{-1}\left(
L_{2}-L_{1}\right) -\left( F_{1}-F_{2}\right) \left( zI-A_{F_{2}}\right)
^{-1}L_{2}.
\end{equation*}%
The lemma is proved.


\end{document}